\documentclass[a4paper,UKenglish,cleveref, autoref, thm-restate]{lipics-v2021}

\usepackage{amsmath}
\usepackage{amsthm}
\usepackage{amsfonts}
\usepackage{mathtools} 
\usepackage{stmaryrd}
\usepackage{cmll}
\usepackage{mathpartir}
\usepackage{tikz,tikz-cd}
\usetikzlibrary{arrows,arrows.meta,positioning}
\usepackage{url}
\usepackage{enumerate,xspace}
\usepackage{comment}
\usepackage[colorinlistoftodos,prependcaption,textsize=tiny]{todonotes}
\usepackage{xcolor}
\usepackage[all]{xy}
\usepackage{braket}
\usetikzlibrary{fit}
\usepackage{graphicx}
\graphicspath{  }
\usepackage{bm} % for \bm
\usepackage{fixmath} % for \mathbold

\newcommand{\ie}{\emph{i.e.}\;}
\newcommand{\eg}{\emph{e.g.}\,}

\newcommand{\id}{\mathrm{id}}
\newcommand{\op}{\mathrm{op}}
\newcommand{\co}{\mathrm{co}}
\newcommand{\cat}[1]{\mathbb{#1}}
\newcommand{\bicat}[1]{\mathscr{#1}}
\newcommand{\catname}[1]{\mathbf{#1}}
\renewcommand{\Set}{\catname{Set}}
\newcommand{\Rel}{\catname{Rel}}

\newcommand{\sem}[1]{\llbracket #1 \rrbracket}

\newcommand{\Pow}{\mathcal{P}} %powerset

\newcommand{\strength}{\mathrm{str}}

\newcommand{\Barr}[1]{\overline{#1}}% Barr extension
\renewcommand{\graph}[1]{\mathbf{gr}{(#1)}}
\newcommand{\coll}[1]{\mathbf{cl}{(#1)}}

\newcommand{\nextt}{\mathrm{next}} 
\newcommand{\out}{\mathrm{out}} 

\newcommand{\vdashc}[1][]{\vdash^{#1}_{c}}
\newcommand{\vdashv}[1][]{\vdash^{#1}_{v}}

\newcommand{\Loc}{\mathrm{Loc}} %memory location
\newcommand{\bool}{\mathbf{bool}}
\newcommand{\unit}{\mathbf{1}}

\newcommand{\Val}{\mathcal{V}}
\newcommand{\letin}[2]{\mathtt{let}\,#1\,\mathtt{in}\, #2}

\newcommand{\return}{\mathtt{return}\,}

\newcommand{\true}{\mathbf{t}}
\newcommand{\false}{\mathbf{f}}

\newcommand{\inj}{\mathrm{in}}

\usetikzlibrary{decorations.markings}
\tikzset{module/.style={%
		decoration={markings,%
			mark= at position 0.5 with {%
				\node[] (tempnode) {};
				\draw[-,shorten <=.1em,shorten >=.1em,inner sep=1pt]
				(tempnode.north) -- (tempnode.south);
			}
		},
		postaction={decorate}
	}
}
\newcommand{\profto}[1][\null]{%
	\mathrel{%
		\tikz{\draw[module,->,transform canvas={yshift=3pt}%
			] (0,0) -- node[above]{$\scriptstyle #1$} ++(1/2,0) ;%
		}
	}%
}

\title{Coinductive reasoning for parametrized functors and monads} 

\author{Ugo Dal Lago}{University of Bologna, Italy}{ugo.dallago@unibo.it}{https://orcid.org/0000-0001-9200-070X}{}
\author{Zeinab Galal}{RIMS, Kyoto University, Japan}{zgalal@kurims.kyoto-u.ac.jp}{https://orcid.org/0009-0008-6402-3531}{}

\authorrunning{U. Dal Lago and Z. Galal} 

\Copyright{Ugo Dal Lago and Zeinab Galal} 


\category{} 

\relatedversion{}

\acknowledgements{The first author was partly supported by the ANR Project HOPR (ANR-24-CE48-5521-01) and the second author was supported by the Research Institute for Mathematical Sciences (RIMS) and by the JSPS KAKENHI Grant 26K21167.}

\nolinenumbers %uncomment to disable line numbering

\EventEditors{Ana Sokolova and Patrick Totzke}
\EventNoEds{2}
\EventLongTitle{37th International Conference on Concurrency Theory (CONCUR 2026)}
\EventShortTitle{CONCUR 2026}
\EventAcronym{CONCUR}
\EventYear{2026}
\EventDate{September 1--4, 2026}
\EventLocation{Liverpool, UK}
\EventLogo{}
\SeriesVolume{391}
\ArticleNo{21}
\begin{document}

\maketitle

\begin{abstract}
Lax extensions (also called relators or relation liftings) are a categorical notion to reason about functors acting on functions and relations in a compatible way. 
They play a central role to develop sound proof principles for behavioral equivalence of state-based systems and are also important for establishing contextual equivalence for effectful programs.

In this paper, we develop the theory of lax extensions for parametrized functors and monads and consider notions of behavioral preorders, equivalence relations or metrics which can now be modulated by additional parameters. 
From an operational viewpoint, we replace standard contextual equivalence where we quantify over all possible contexts by a refined notion of equivalence where the user can regulate the allowed contexts via chosen parameters.
\end{abstract}
\paragraph*{Introduction}

Relations play a central role in modeling various notions of behavioral preorders and equivalences for state-based systems and programming languages. In particular, coinductive techniques on relations provide a formal framework to reason about circular or infinite objects and they are an important tool in concurrency theory and in operational semantics for establishing behavioral equivalence of processes or programs~\cite{sangiorgi2011introduction}. Simulation and bisimulation are among the main examples of relations defined coinductively and they have a categorical counterpart in the world of coalgebras. Coalgebras are morphisms $S \longrightarrow F(S)$ where $S$ represents the state space and $F$ is a functor encoding the branching type of the system~\cite{rutten2000universal,jacobs2017introduction}. By varying the functor, we can capture a wide range of state-based systems: streams, labelled transition systems, automata, Markov chains, etc.

A fundamental component of the study of behavioral preorders and equivalences is to extend the behavior functor of the system from sets to relations using the notion of lax extension.
This notion also plays an important role in the operational semantics of higher-order functional languages with monadic effects where applicative bisimulation is generalized to the effectful setting using lax extensions of monads from sets to 
relations~\cite{dal_lago_relational_2022}.

Lax extensions are also used in modal logic, effectful intersection type systems, monoidal topology and in the theory of double categories~\cite{marti2015lax,gavazzo2024monadic,Hofmann_Seal_Tholen_2014,lambert2021characterizing}. They were initially considered by Barr for the ultrafilter monad on sets and were since then generalized to quantale-enriched relations~\cite{barrrelational,rutten1998relators,worrell2000coalgebras,seal2005canonical,schubert2008extensions,bilkova2013relation} to measure the \emph{behavioral distance} between two states for non-deterministic and probabilistic systems~\cite{van2005behavioural,baldan2018coalgebraic,konig2018metric}.

The main purpose of this paper is to extend the theory of lax extensions for 
endofunctors and monads to \emph{parametrized} functors and monads.
A first generalization is to replace endofunctors $F : \cat{C} \to \cat{C}$ with functors $F : \cat{P} \times \cat{C} \to \cat{C}$ where $\cat{P}$ is viewed as a parameter category.
 This simple generalization allows us to vary the parameter and consider notions of behavioral preorders or equivalences which depend on relations on the parameter.
 
In the case of monads, the underlying endofunctor of an ordinary monad $T :  
\cat{C} \to \cat{C}$ is replaced by a functor $T:  \cat{P}^\op \times \cat{P} 
\times \cat{C} \to \cat{C}$~\cite{atkey2009parameterised}. A parametrized monad has therefore two additional 
parameters providing information about the starting and ending state of a 
computation. It is important to note that the input parameter is contravariant 
and the output parameter is covariant, following the intuition of preconditions 
and postconditions for program logics. Parametrized monads were introduced by Atkey~\cite{atkey2009parameterised} and they offer a convenient setting to model sequential effectful computations, they also have a categorical 
characterization as the notion of monads arising from parametrized 
adjunctions~\cite{mac1998categories,Cheng_Gurski_Riehl_2014} (also called adjunction with a parameter or multivariable adjunction) offering the natural generalization of the standard 
monad-adjunction correspondence.
One of the motivating examples is the global state monad: in the ordinary 
case, the set of global stores is fixed and remains constant which is not 
flexible enough to keep track of which memory regions a computation has access to and can modify. In the parametrized case, the change in the state type over 
time is captured by additional parameters corresponding to the memory 
region \emph{before} and \emph{after} the execution of a computation.

In this paper, we develop the theory of lax extensions for parametrized monads to have a framework that allows us to modulate the notion of program equivalence or metric depending on the power we want to give to contexts. 
For example, if programs can read and write on some memory locations, standard contextual equivalence means that contexts have access to the same memory locations as programs. We are interested in having well-behaved notions of program equivalence where contexts only have access to a \emph{restricted} set of memory locations. In this setting, our framework allows to weaken the notion of equivalence by explicitly controlling the power we give to contexts depending on the relations added to parameters.

\subsection*{Plan of the paper}

\begin{itemize}
	\item After a background section on parametrized monads (Section~\ref{sec:paramMonads}), we introduce the notion of parametrized lax extensions for parametrized functors and monads in Section~\ref{sec:paramLaxExt}.
	\item In Section~\ref{sec:Barr}, we generalize the Barr canonical extension to the parametrized setting and show how additional nuances need to be addressed in the presence of mixed variance variables. In particular, there are several options for the Barr extension, and we study their properties.
	\item In Section~\ref{sec:bisim}, we study notions of bisimulation and behavioral equivalence in the parametrized setting where the behavior of two states is compared relative to some chosen relations on the parameters. 
	\item Lastly, we instantiate our framework in Section~\ref{sec:calculus} to coalgebras for parametrized functors in the first order setting (Mealy machines) and to parametrized monadic higher order languages where we introduce a parametrized version of contextual equivalence.
\end{itemize}

\section{Background on parametrized monads}\label{sec:paramMonads}
We begin by recalling the notion of parametrized (strong) monad and provide some examples. While the term ``parametrized'' has been used in various contexts, we would like to emphasize that the notion of parametrized monad we consider in this paper is the one by Atkey~\cite{atkey2009parameterised} and is distinct from the other generalizations of monads considered in the literature such as indexed monads~\cite{fujii20192}, parametric eﬀect monads~\cite{katsumata2014parametric}, Dijkstra monads~\cite{maillard2019dijkstra}, graded monads~\cite{smirnov2008graded}, etc.

\begin{definition}[Definitions 5.2.7 and 5.2.9 in~\cite{atkey2006substructural}]\label{def:paramMonads}
	For categories $\cat{C}$ and $\cat{D}$, a \emph{parametrized monad} consists of a functor $T : \cat{C}^\op \times \cat{C} \times \cat{D} \to \cat{D}$ together with 
	\begin{itemize}
		\item a family of morphisms $\eta_{S,X}: X \to T(S,S,X)$ indexed by objects $S$ in $\cat{C}$ and $X$ in $\cat{D}$ natural in $X$ and dinatural in $S$
		\item a family of morphisms $\mu_{S_1 ,S_2 ,S_3 ,X}	: T(S_1,S_2,T(S_2,S_3,X))\to T(S_1,S_3,X)$ indexed by objects $S_1,S_2,S_3$ in $\cat{C}$ and $X$ in $\cat{D}$ natural in $X, S_1,S_3$ and dinatural in $S_2$
	\end{itemize}
	satisfying some compatibility axioms.
	If $(\cat{D}, \otimes, 1)$ is monoidal, a \emph{parametrized strong monad} is additionally equipped with 
	\begin{itemize}
 		\item a family of morphisms $\strength_{S_1 ,S_2 ,X,Y}	: X \otimes T(S_1, S_2,Y) \to  T(S_1,S_2,X \otimes Y)$ in $\cat{D}$ indexed by objects $S_1,S_2$ in $\cat{C}$ and $X,Y$ in $\cat{D}$, natural in all variables, and satisfying some compatibility axioms.
	\end{itemize}
\end{definition}

\begin{example}\label{ex:paramStateMonad}
	\begin{enumerate}
		\item The parametrized state monad~\cite[Section 2.3.2]{atkey2009parameterised}: for a fixed set of states $S$, the global state monad maps a set $X$ to the set $(X \times S)^S$ while the parametrized version allows for the set of input states and output states to vary, its underlying functor $T : \Set^\op \times \Set \times \Set \to \Set$ has the following action on objects: $(S_1, S_2, A) \mapsto (A \times S_2)^{S_1}$.
		\item The composable continuation monad \cite[Example 5.2.6]{atkey2006substructural}:
		for a fixed response object $R$ in a cartesian closed category, the ordinary continuation monad maps an object $X$ to $(X \Rightarrow R) \Rightarrow R$. In the parametrized case, we obtain the notion of composable continuations where the response type may vary: the composable continuation monad $T : \Set^\op \times \Set \times \Set \to \Set$ is given by $(R_1, R_2, X) \mapsto 
		(X \Rightarrow R_1)\Rightarrow R_2$.
		\item The parametrized I/O monad~\cite[Section 2.3.4]{atkey2009parameterised}: for a fixed set of inputs $I$ and observable outputs $O$, the interactive input-output (I/O) monad  is the functor mapping a set $X$ to the underlying set of the initial algebra $ \mu Z. (X + Z^I+ O\times Z)$. It induces three possible types of computations: either compute a value, wait for an input and then resume or produce an output and then resume~\cite[Example 1.1]{moggi1991notions}. 
		The parametrized $I/O$-monad $\Set^\op \times \Set \times \Set \to \Set$ maps a triple $(I, O, X)$ to the underlying set of $\mu Z. (X + Z^I + O\times Z)$.
		\item The parametrized action monad \cite[Example 5.2.5]{atkey2006substructural}: any monoid $(M, e: I \to M, m : M \otimes M \to M)$ internal to a monoidal category $\cat{C}$ induces a monad (called the writer, output or action monad) on $\cat{C}$ given by $X \mapsto M \otimes X$. This construction has been generalized by Atkey in the case of parametrized monads where any category $\cat{P}$ induces a parametrized monad $\cat{P}^\op \times \cat{P} \times \Set \to \Set$ whose action on objects is $(A,B, X) \mapsto \cat{P}(A,B) \times X$.
	\end{enumerate}
\end{example}

\section{Parametrized lax extensions}\label{sec:paramLaxExt}

\subsection{Relations and dualities}
The simplest notion of relation is a binary relation between sets $R \subseteq X \times Y$ which we denote $R : X \profto Y$. The identity relation $\id_X : X \profto X$ is the diagonal $\{(x,x) \mid x \in X\}$ and the composite of two relations $R: X \profto Y$ and $S : Y \profto Z$ is the relation $S \circ R : X \profto Z$ given by $\{(x,z) \mid \exists y \in Y (x,y) \in R \text{ and } (y,z) \in S\}$. We will also use the sequential notation $R ; S$ for the composite $S \circ R$ and we denote by $\Rel$ the category of sets and binary relations.

The underlying functor of a parametrized monad $T : \Set^\op \times \Set \times \Set \to \Set$ has both contravariant and covariant inputs so it is important to understand the dualities it induces on relations to motivate the definition of parametrized extension. They are better understood in light of the $2$-categorical structure of $\Rel$: for two relations  $R,S : X \profto Y$,  the inclusion $R \subseteq S$ viewed as subsets of $X \times Y$ induces an ordering on hom-sets $(\Rel(X,Y), \leq)$. Composition of relations is monotone with respect to this ordering which implies that $\Rel$ is a posetal $2$-category. 

For a general $2$-category $\bicat{C}$, we have the ``$\op$'' operator which reverses the direction of $1$-morphisms and we also have the ``$\co$'' operator which reverses the direction of $2$-morphisms. For example, the transpose operator that maps a relation $R : X\profto Y$ to the relation $R^t =\{(y,x) \mid (x,y) \in R\} : Y \profto X$ can be seen as a $2$-functor $(-)^t : \Rel^\op \to \Rel$.

The exponentiation operation in $\Set$ is a functor $\Set^\op \times \Set \to \Set$ mapping a pair $(X,Y)$ to $Y^X$, the set of all functions from $X$ to $Y$. It induces an exponentiation\footnote{While $\Rel$ is not a cartesian closed category, we still use the terminology of exponentiation because $\Rel$ is cartesian closed as a double category~\cite{niefield2024cartesian}.} $2$-functor $\Rel^\co \times \Rel \to \Rel$ that maps a pair of relations $(R: X_1 \profto X_2, S : Y_1  \profto Y_2)$ to \[S^R := \{ (\varphi, \psi) \in Y_1^{X_1} \times Y_2^{X_2} \mid \forall (x_1, x_2) \in R, (\varphi(x_1), \psi(x_2)) \in S
\}\] which reverses the inclusion of relations in the first argument, \ie $R \leq R'$ implies $S^{R'} \leq S^{R}$. This operation will guide the formulation of the theory of parametrized lax extensions, and the exponentiation of relations will provide the canonical extension for the state monad.

Another important fact for the theory of lax extensions is the inclusion of functions into relations: any function $f : X\to Y$ induces two relations $f_* : X \profto Y$ and $f^* : Y \profto X$ given by $f_*:= \{ (x, f(x)) \mid x\in X\}$ and $f^*:= \{ (f(x), x) \mid x\in X\}$.
Moreover, $f_*$ is left adjoint to $f^*$ in (the posetal $2$-category) $\Rel$.

\subsection{Parametrized (lax) extensions}

In the non-parametrized setting, for an endofunctor $F: \Set \to \Set$, a \emph{strict extension} (or simply an extension) of $F$ to $\Rel$ consists of an endofunctor $\Barr{F}: \Rel \to \Rel$ that agrees with $F$ on objects ($F(A)= \Barr{F}(A)$) and that is compatible with the inclusion of functions into relations, \ie for a function $f: A \to B$, $(F(f))_* = \Barr{F}(f_*)$. An extension $\Barr{F}$ is said to be \emph{monotone} if it preserves the inclusion of relations: if $R \leq R'$, then $\Barr{F}(R) \leq \Barr{F}(R')$. Equivalently,  $\Barr{F}$ is a $2$-functor $\Rel \to \Rel$. 

The general idea is that, in the strict case, a monotone extension of a functor $T : \Set^\op \times \Set \times \Set \to \Set$ will be a $2$-functor $\Barr{T} : \Rel^\co \times \Rel \times \Rel \to \Rel$ and not $\Barr{T} : \Rel^\op \times \Rel \times \Rel \to \Rel$.

\begin{definition}\label{def:strictParamExtension}
	For a functor $T : \Set^\op \times \Set \times \Set \to \Set$, a $1$-functor $\Barr{T} :  \Rel \times \Rel \times \Rel \to \Rel$ extends $T$ strictly if the following holds:
	\begin{enumerate}
		\item for all $f : A \to B, g : C\to D$ and $h : X\to Y$, $\Barr{T}(f^*,g_*,h_*)=(T(f,g,h))_*$
\item we say that $\Barr{T}$ is \emph{monotone} if for all $M'\leq M : A \profto B$, $N\leq N': C \profto D$ and $R\leq R' : X \profto Y$, $\Barr{T}(M ,N, R) \leq \Barr{T}(M' ,N',R')$.
	\end{enumerate}
\end{definition}

Strict extensions are in practice too restrictive and various weakenings have been considered in the non-parametrized setting. In this paper, we focus on the notion of \emph{lax extension} considered by Barr and later generalized to the quantale-enriched setting~\cite{barrrelational,worrell2000coalgebras,seal2005canonical,schubert2008extensions,bilkova2011relation,bilkova2013relation,double}.

\begin{definition}\label{def:laxextension}
	For a functor $F: \Set \to \Set$, a \emph{lax extension} of $F$ to $\Rel$ consists of an indexed family of functions $\Barr{T}:\Rel(X,Y) \longrightarrow \Rel(TX, TY)$ satisfying the following axioms:
	\begin{enumerate}
		\item $\Barr{T}$ is monotone: if $R \leq R' : X \profto Y$, then $\Barr{T}(R) \leq \Barr{T}(R')$;
		\item for composable relations $R : X \profto Y$, $S : Y \profto Z$, $\Barr{T}(S) \circ \Barr{T}(R) \leq \Barr{T}(S \circ R)$;
		\item for all functions $f: X \to Y$;
		$(T(f))_* \leq \Barr{T}(f_*)$ and $(T(f))^* \leq \Barr{T}(f^*)$.
	\end{enumerate}
	If $\Barr{T}$ preserves identities strictly, \ie $\id_{TX} = \Barr{T}(\id_X)$, we say that it is a \emph{normal} lax extension (also called \emph{flat}). If conditions $2.$ and $3.$ are equalities, then we recover exactly the notion of monotone strict extension.
\end{definition}

Lax extensions are not unique, there can be more than one lax extension for the same functor:
\begin{example}
	In the case of the power set functor $\Pow: \Set \to \Set$, for a relation $R : X \profto Y$, the following mappings
	\[\begin{aligned}
		\mathcal{L}_1(R)&:=\{ (A,B) \in \Pow(X)\times \Pow(Y) \mid \forall x \in A, \exists y \in B, (x,y) \in R\}\\
		\mathcal{L}_2(R)&:=\{ (A,B) \in \Pow(X)\times \Pow(Y) \mid \forall y \in B, \exists x \in A, (x,y) \in R\}\\
		\mathcal{L}_3(R)&:=	\mathcal{L}_1(R) \cap 	\mathcal{L}_2(R)
	\end{aligned}\]
	are all lax extensions. The last one corresponds to the canonical Barr extension in Section~\ref{subsec:ordinaryBarr}.
	\end{example}

In the parametrized case, we are also interested in the lax version where the strict equalities in Definition~\ref{def:strictParamExtension} are relaxed to inequalities.

\begin{definition}\label{def:ParamLaxExtensionMonadSet}
	For a functor $T: \Set^\op \times \Set \times \Set \to \Set$, a \emph{parametrized lax extension} consists of a family of functions: $\Barr{T}:	\Rel(A,B) \times \Rel(C,D) \times \Rel(X, Y)  \to \Rel(T(A, C,X), T(B,D,Y))$
	indexed by $A,B,C,D,X,Y$ that satisfies the following axioms:
	\begin{enumerate}
		\item it is antitone in the first argument and monotone in the last two arguments:  for all $M'\leq M : A \profto B$, $N\leq N': C \profto D$ and $R\leq R' : X \profto Y$, 
		\[\Barr{T}(M ,N, R) \leq \Barr{T}(M' ,N',R')\]
		\item lax composition: for all $M: A_1 \profto A_2$, $M' :A_2 \profto A_3$,  $N: B_1 \profto B_2$, $N' :B_2 \profto B_3$, $R: X_1 \profto X_2$ and $R' :X_2 \profto X_3$, 
		\[\Barr{T}(M,N,R)  ; \Barr{T}(M',N',R') \leq  	\Barr{T}(M; M',N; N',R; R')\]
		\item for all $f : A\to B$, $g : C \to D$ and $h : X \to Y$, 
		\[(T(f, g, h))_* \leq \Barr{T}(f^*, g_*, h_* )\text{ and }(T(f, g, h))^*\leq \Barr{T}(f_*,  g^*, h^* ).\]
	\end{enumerate}
		If $\Barr{T}$ preserves identities strictly $(\id_{T(A,B,X)} = \Barr{T}(\id_{A}, \id_{B},\id_X))$, we say that it is a \emph{normal} parametrized lax extension. If conditions $2.$ and $3.$ are equalities, then we obtain a monotone strict extension.
		If $\Barr{T}$ commutes with the transpose operation as below, it is called \emph{symmetric}:
		\begin{enumerate}
			    \setcounter{enumi}{3}
			\item for all $M : A \profto B$, $N: C \profto D$ and $R : X \profto Y$, $(\Barr{T}(M ,N, R))^t = \Barr{T}(M^t ,N^t,R^t)$
		\end{enumerate}
\end{definition}

\begin{example}\label{ex:statenostrictextension}
The underlying functor of the parametrized state monad (Example~\ref{ex:paramStateMonad}.1) has a parametrized lax extension $\Barr{T}$ which we describe explicitly: for relations $M : A \profto B$, $N: C \profto D$ and $R : X \profto Y$, $\Barr{T}(M,N,R): (C \times X)^A \profto (D \times Y)^B$ is given by:
\[
(N \times R)^M:=\{ (\varphi, \psi) \mid \forall(a,b) \in M, (\varphi(a), \psi(b)) \in N \times R\}.
\]
\end{example}

\section{The Parametrized Barr extensions}\label{sec:Barr}

\subsection{Extending functors from $\Set$ to $\Rel$}\label{subsec:ordinaryBarr}
 
\paragraph*{The span representation of relations}
The Barr extension (also called canonical relation lifting) is based on the universal representation of relations as pairs of functions which we recall here.
A relation $R : X \profto Y$ induces a pair of projection functions $\pi^R_1 : \graph{R} \to X$ and $\pi^R_2 :  \graph{R} \to Y$ where $\graph{R}:=\{(x,y) \in X \times Y \mid (x,y) \in R\} \leq X \times Y$ is called the \emph{graph} or \emph{tabulator} of $R$. These two functions completely describe the relation in the following sense:
\begin{enumerate}
	\item tabulator universal property: for every other span of functions $(h : A \to X, k : A \to Y)$ satisfying $ h^* ; k_* \leq R$, there exists a unique function $u : A \to  \graph{R}$ such that $\pi_1^R \circ u =h$ and $\pi_2^R \circ u =k$.
	\item tabulators are strong: we can recover the original relation $R =  (\pi^R_1)^* ; ( \pi^R_2)_*$.
\end{enumerate}

\paragraph*{The cospan representation of relations}
There is a dual cospan representation of relations with the notion of \emph{collage} or \emph{cotabulator}. It is used to define the Barr extension of quantale-enriched relations, we will only define it for ordinary relations on sets and make use of its universal property in Section~\ref{sec:bisim}. For a relation $R : X \profto Y$, if we take the pushout of the span $(\pi^R_1 : \graph{R} \to X,\pi^R_2 :  \graph{R} \to Y)$, we obtain a pair of functions $\inj^R_1 : X \to \coll{R}$ and $\inj^R_2 : Y \to \coll{R}$ where $\coll{R}:= X+Y/_\sim$ is the quotient of the disjoint union $X + Y$ by the equivalence relation $\sim$ generated by $(2,y) \, \sim \, (1,x)$ for $(x,y) \in R$. It verifies the following:
	\begin{enumerate}
	\item cotabulator universal property: for every other cospan of functions $(h : X \to A , k : Y \to A)$ satisfying $ R \leq h_*;  k^* $, there exists a unique function $u : \coll{R} \to A$ such that $u\circ \inj_1^R =h$ and $ u\circ \inj_1^R =k$.
\end{enumerate}
 $\Rel$ does not have strong cotabulators, we only have $R \leq  (\inj^R_1)_*;( \inj^R_2)^* $ in general. Relations for which the cotabulator is strong are called \emph{difunctional} and they subsume equivalence relations~\cite{riguet1948relations,gumm2014coalgebraic,goncharovstacs25}.

\paragraph*{The Barr extension}

The representation of a relation by a universal span of functions is at the heart of the Barr extension construction. The Barr extension $\Barr{T}:\Rel(X,Y) \longrightarrow \Rel(TX, TY)$ is the operator
mapping a relation $R : X \profto Y$ to the relation $\Barr{T}(R):= (T\pi^R_1)^* ; (T\pi^R_2)_*: TX \profto TY$.

Without any additional assumptions on $T$, the operator $\Barr{T}$ is monotone (axiom $(1)$ of Definition~\ref{def:laxextension}) and satisfies axiom $(3)$ strictly. 
However, it is not in general a lax extension since it does not always satisfy the lax composition axiom $(2)$, it preserves composition in the reverse op-lax direction instead. 
To obtain a lax-extension, we need to assume additional properties on $T$ based on the notion of \emph{weak pullback}:
\begin{definition}\label{def:weakpb}
	A commutative square $h\circ f= k \circ g$ in $\Set$ as below is a weak pullback if the universal map from $A$ to the pullback of $h$ and $k$ is split epi.
	\vspace{-0.1cm}
	\begin{center}
		\begin{tikzpicture}[thick,xscale=0.5, yscale=0.6]
			\node (A) at (0,1.5) {\scriptsize$B$};
			\node (B) at (2,0) {\scriptsize$D$};
			\node (C) at (4,1.5) {\scriptsize$C$};
			\node (D) at (2,3.5) {\scriptsize$A$};
			\node (P) at (2,2.25) {\scriptsize$P$};
			\draw [-to, dotted] (D)  to  node [below left] {} (P);
			\draw [->] (A)  to  node [below left] {\scriptsize$h$} (B);
			\draw [->] (C)  to  node [below right] {\scriptsize$k$} (B);
			\draw [->] (D) to  [bend left =-33] node [above left] {\scriptsize$f$} (A);
			\draw [->] (D)  to [bend left =33] node [above right] {\scriptsize$g$} (C);
			\draw [->] (P) to  node [below] {\scriptsize$\pi_1$} (A);
			\draw [->] (P)  to node [below] {\scriptsize$\pi_2$} (C);
			\node at (2,1.25) {\scriptsize$\mathrm{pb}$};
		\end{tikzpicture}
		\vspace{-0.3cm}
	\end{center}
\end{definition}
This notion (also called Beck-Chevalley property) plays a central role in the theory of lax extensions as weak pullback squares can be equivalently characterized by the equality $g_* \circ f^* = k^* \circ h_*$ in $\Rel$ allowing to commute the left $(-)_*$ and right adjoints $(-)^*$.

\begin{theorem}[\cite{barrrelational,trnkova1977relational,carboni19912}]\label{Barrtheorem}
	For a functor $T : \Set \to \Set$, $T$ preserves weak-pullbacks if and only if the Barr extension is a monotone strict extension and it is the unique one.
\end{theorem}

Before considering the extension of functors $T : \Set^\op \times \Set \times \Set \to \Set$, we first explain how to extend functors $ \Set^\op \to \Set$ and $\Set \times \Set \to \Set$ to isolate the different issues that each generalization entails. 

\subsection{Extending contravariant functors}

For a functor $T:\Set^\op \to \Set$, we define its Barr extension as the operator mapping a relation $R : X \profto Y$ to the relation $\Barr{T}(R):= (T\pi^R_1)_* ;(T\pi^R_2)^*: TX \profto TY$.
The operator $\Barr{T}$ is antitone with respect to the inclusion of relations, is normal and verifies $\Barr{T}(f^*) = (Tf)_*$ and $\Barr{T}(f_*) = (Tf)^*$ for a function $f$. 
The important point is that it always preserves composition of relations laxly without requiring additional conditions on $T$. To obtain a strict extension, it suffices for $T$ to map weak pullback squares to weak pullback squares in the following sense:
	\begin{center}
	\begin{tikzpicture}[thick,scale=0.5]
		\node (X) at (-2,0) {};
		\node (A) at (0,1.5) {\scriptsize$P$};
		\node (B) at (2,0) {\scriptsize$B$};
		\node (C) at (4,1.5) {\scriptsize$C$};
		\node (D) at (2,3) {\scriptsize$A$};
		
		\draw [->] (A)  to  node [below left] {\scriptsize$p_2$} (B);
		\draw [->] (B)  to  node [below right] {\scriptsize$g$} (C);
		\draw [->] (A) to  node [above left] {\scriptsize$p_1$} (D);
		\draw [->] (D)  to node [above right] {\scriptsize$f$} (C);
		\node at (2,1.5) {\scriptsize$\mathrm{wpb}$};
		\node at (5.5,1.5) {$\Rightarrow$};
		
		\begin{scope}[xshift=7cm]
				\node (A) at (0,1.5) {\scriptsize$TC$};
			\node (B) at (2,0) {\scriptsize$TB$};
			\node (C) at (4,1.5) {\scriptsize$TP$};
			\node (D) at (2,3) {\scriptsize$TA$};
			
			\draw [->] (A)  to  node [below left] {\scriptsize$T(g)$} (B);
			\draw [->] (B)  to  node [below right] {\scriptsize$T(p_2)$} (C);
			\draw [->] (A) to  node [above left] {\scriptsize$T(f)$} (D);
			\draw [->] (D)  to node [above right] {\scriptsize$T(p_1)$} (C);
			\node at (2,1.5) {\scriptsize$\mathrm{wpb}$};
		\end{scope}
	\end{tikzpicture}
\end{center}
Note however that this condition is quite strong and that the typical examples built from the functor $(-)\Rightarrow A: \Set^\op \to \Set$ do not satisfy this property as the exponentiation functor is only lax in $\Rel$. 
\subsection{Extending multivariable functors}
We consider here functors of the form $T: \Set \times \Set \to \Set$ with two arguments but this can be generalized to functors $\Set^n \to \Set$. We define the \emph{parametrized Barr extension} of $T$ as the operator $\Barr{T}$ mapping a pair of relations $(N: C\profto D, R: X\profto Y)$ to $\Barr{T}(N,R):= (T( \pi^N_1, \pi^R_1))^* ; T(\pi^N_2, \pi^R_2))_* : T(C,X)\profto T(D,Y)$.
It verifies the expected properties which we do not prove as they are a special case of Theorem~\ref{th:paramBarr} where we have an additional contravariant variable. Instead, we want to emphasize the fact that in the multivariable case, we have other possibilities for the choice of extensions we can define:
\[\begin{aligned}
	\mathcal{L}_1&(N,R):= (T( \pi^N_1, \id))^*; (T(\pi^N_2, \id ))_* ; (T( \id, \pi^R_1))^*;(T(\id, \pi^R_2))_* \\
	\mathcal{L}_2&(N,R):= (T( \id, \pi^R_1))^*;(T(\id, \pi^R_2))_* ; (T( \pi^N_1, \id))^*; (T(\pi^N_2, \id ))_* \\
\end{aligned}\]
We have that $\Barr{T}(N,R) \leq \mathcal{L}_1(N,R)$ and $\Barr{T}(N,R) \leq \mathcal{L}_2(N,R)$. While $\Barr{T}$ preserves composition oplaxly in both variables, $\mathcal{L}_1$ and $\mathcal{L}_2$ preserve composition oplaxly in each variable separately: $\mathcal{L}_i(N'\circ N, \id) \leq \mathcal{L}_i(N', \id)\circ \mathcal{L}_i( N, \id)$ and $ \mathcal{L}_i(\id, R'\circ R)\leq \mathcal{L}_i(\id, R')\circ \mathcal{L}_i(\id, R) $.

Recall that for an oplax $2$-functor, we have both $ T(f,g)\leq T(f, \id) \circ T(\id, g) $ and $  T(f,g)\leq T(\id, g) \circ T(f, \id)$ but $T(f, \id) \circ T(\id, g)$ and $T(\id, g) \circ T(f, \id)$ are not in general comparable. For multivariable oplax $2$-functors, the equality $T(f, \id) \circ T(\id, g) = T(\id, g) \circ T(f, \id)$ (or coherent isomorphism if we are not in the posetal case) corresponds to the intermediate notion of cubical or quasi-functors~\cite{gray2006formal,johnson20212}) between strict (or pseudo) $2$-functors and oplax ones.
 If for every $f: A \to B$, $g : C \to D$, the following interchange square is a weak-pullback
\begin{center}
	\begin{tikzpicture}[thick,scale=0.5]
		
		\node (A) at (0,1.5) {\scriptsize$T(B,C)$};
		\node (B) at (2,0) {\scriptsize$T(B,D)$};
		\node (C) at (4,1.5) {\scriptsize$T(A,D)$};
		\node (D) at (2,3) {\scriptsize$T(A,C)$};
		
		\draw [->] (A)  to  node [below left] {\scriptsize$T(\id, g)$} (B);
		\draw [->] (C)  to  node [below right] {\scriptsize$T(f, \id)$} (B);
		\draw [->] (D) to  node [above left] {\scriptsize$T(f, \id)$} (A);
		\draw [->] (D)  to node [above right] {\scriptsize$T(\id, g)$} (C);
		\node at (2,1.5) {\scriptsize$\mathrm{wpb}$};
		
	\end{tikzpicture}
\end{center} 
then $\Barr{T} = \mathcal{L}_1= \mathcal{L}_2$ which intuitively means that we can commute the order in which we do the Barr extension. 
This issue will be relevant in the mixed variance case where we have to consider the covariant and contravariant variables separately.

\subsection{Combining both}\label{sec:paramBarrextRel}
For a functor $T : \Set^\op \times \Set \times \Set \to \Set$, we obtain a lattice of \emph{parametrized Barr extensions}: for a triple of relations $(M : A\profto B, N: C\profto D, R: X\profto Y)$, we define
\[\begin{aligned}
	\mathcal{L}_1(M,N,R)&:=(T(\id, \pi_1, \pi_1))^* ;(T(\pi_1, \id, \id))_* ;(T(\pi_2, \id, \id))^*;(T(\id, \pi_2, \pi_2))_* \\
	\mathcal{L}_2(M,N,R)&:= (T(\pi_1, \id, \id))_* ;T(\pi_2, \pi_1, \pi_1))^* ;(T(\id, \pi_2, \pi_2))_*\\
	\mathcal{L}_3(M,N,R)&:= (T(\id, \pi_1, \pi_1)^* ; (T( \pi_1, \pi_2, \pi_2))_*;(T(\pi_2, \id, \id))^*\\
	\mathcal{L}_4(M,N,R)&:= (T(\pi_1, \id, \id))_* ;(T(\id, \pi_1, \pi_1))^* ;(T(\id, \pi_2, \pi_2))_* ;(T(\pi_2, \id, \id))^*
\end{aligned}\]
These operators are related by the following inclusions: $\mathcal{L}_1 \leq\mathcal{L}_2 \leq\mathcal{L}_4$ and $\mathcal{L}_1 \leq\mathcal{L}_3\leq\mathcal{L}_4$.

\begin{theorem}\label{th:paramBarr}
	For a functor $T$ as above and $i\in \{1,2,3,4\}$, the Barr extension $\mathcal{L}_i$ verifies:
\begin{enumerate}
		\item it is normal: for all $A,B,X$, $\id_{T(A,B,X)} = \mathcal{L}_i(\id_{A}, \id_{B},\id_X)$ 
	\item it is antitone in the first argument and monotone in the last two arguments: for all $M'\leq M : A \profto B$, $N\leq N': C \profto D$ and $R\leq R' : X \profto Y$, 
	\[\mathcal{L}_i(M ,N, R) \leq \mathcal{L}_i(M' ,N',R')\]
	\item it preserves composition laxly in the first argument: for all $M: A_1 \profto A_2$, $M' :A_2 \profto A_3$, 
	\[	\mathcal{L}_i(M,\id, \id) ; \mathcal{L}_i(M',\id, \id)  \leq  	\mathcal{L}_i(M; M',\id, \id)\]
	\item it preserves composition op-laxly in the last two arguments: for all 
	$N: B_1 \profto B_2$, $N' :B_2 \profto B_3$, $R: X_1 \profto X_2$ and $R' :X_2 \profto X_3$, 
	\[\mathcal{L}_i(\id ,N; N',R; R') \leq \mathcal{L}_i(\id ,N,R) ; \mathcal{L}_i(\id ,N',R')\]
	\item for all $f : A\to B$, $g : C \to D$ and $h : X \to Y$, 
	\[(T(f, g, h))_* = \mathcal{L}_i(f^*, g_*, h_* )\text{ and }(T(f, g, h))^*=  \mathcal{L}_i(f_*,  g^*, h^* )\]
	\item it is symmetric: for all $M : A \profto B$, $N: C \profto D$ and $R : X \profto Y$, 
	\[(\mathcal{L}_i(M ,N, R))^t = \mathcal{L}_i(M^t ,N^t,R^t)\]
	\item it is minimal: for any lax extension $\mathcal{L}$ of $T$ and relations $M,N,R$, we have \[\mathcal{L}_i(M,N,R) \leq \mathcal{L}(M,N,R).\]
\end{enumerate}
\end{theorem}

To obtain that $\mathcal{L}_i$ preserves composition laxly in the last two arguments, we need the functor $T$ to preserve weak pullbacks in the last two arguments. 

\begin{lemma}\label{lem:WeakPbkLastTwo} For $i\in \{1,2,3,4\}$, if for every weak pullback squares as on the left below,
		\begin{center}
		\begin{tikzpicture}[thick,scale=0.5]
			\node (A) at (0,1.5) {\scriptsize$B$};
			\node (B) at (2,0) {\scriptsize$D$};
			\node (C) at (4,1.5) {\scriptsize$C$};
			\node (D) at (2,3) {\scriptsize$A$};
			
			\draw [->] (A)  to  node [below left] {\scriptsize$h$} (B);
			\draw [->] (C)  to  node [below right] {\scriptsize$k$} (B);
			\draw [->] (D) to  node [above left] {\scriptsize$f$} (A);
			\draw [->] (D)  to node [above right] {\scriptsize$g$} (C);
			\node at (2,1.5) {\scriptsize$\mathrm{wpb}$};
			
			\begin{scope}[xshift=5cm]
				\node (A) at (0,1.5) {\scriptsize$X$};
				\node (B) at (2,0) {\scriptsize$Z$};
				\node (C) at (4,1.5) {\scriptsize$Y$};
				\node (D) at (2,3) {\scriptsize$W$};
				
				\draw [->] (A)  to  node [below left] {\scriptsize$l$} (B);
				\draw [->] (C)  to  node [below right] {\scriptsize$m$} (B);
				\draw [->] (D) to  node [above left] {\scriptsize$i$} (A);
				\draw [->] (D)  to node [above right] {\scriptsize$j$} (C);
				\node at (2,1.5) {\scriptsize$\mathrm{wpb}$};
				
					\node at (6,1.5) {$\Rightarrow$};
			\end{scope}

				\begin{scope}[xshift=14cm,yscale=1.2, xscale=1.1]
				\node (A) at (0,1.5) {\scriptsize$T(U,B,X)$};
				\node (B) at (2,0) {\scriptsize$T(U,D,Z)$};
				\node (C) at (4,1.5) {\scriptsize$T(U,C,Y)$};
				\node (D) at (2,3) {\scriptsize$T(U,A,W)$};
				
				\draw [->] (A)  to  node [below left] {\scriptsize$T(\id, h,l)$} (B);
				\draw [->] (C)  to  node [below right] {\scriptsize$T(\id, k.m)$} (B);
				\draw [->] (D) to  node [above left] {\scriptsize$T(\id, f,i)$} (A);
				\draw [->] (D)  to node [above right] {\scriptsize$T(\id, g,j)$} (C);
				\node at (2,1.5) {\scriptsize$\mathrm{wpb}$};
			\end{scope}
		\end{tikzpicture}
	\end{center} their image by $T$ as on the right is a weak pullback, then for all 
	$N: B_1 \profto B_2$, $N' :B_2 \profto B_3$, $R: X_1 \profto X_2$ and $R' :X_2 \profto X_3$, $ \mathcal{L}_i(\id ,N',R') \circ \mathcal{L}_i(\id ,N,R)  \leq	\mathcal{L}_i(\id ,N'\circ N,R'\circ R)$.
\end{lemma}
However, this is not sufficient to obtain a parametrized lax extension as in Definition~\ref{def:ParamLaxExtensionMonadSet}. 
The important point is the question of interchange of variables we explained in the previous section when dealing with two covariant variables. Here, we consider the interaction between the contravariant variable and the covariant ones. We do not obtain 
\[\mathcal{L}_i(M' ,N',R') \circ \mathcal{L}_i(M ,N,R)  \leq	\mathcal{L}_i(M' \circ M ,N'\circ N,R'\circ R)\] from the inclusions $\mathcal{L}_i(\id ,N',R') \circ \mathcal{L}_i(\id ,N,R)  \leq	\mathcal{L}_i(\id ,N'\circ N,R'\circ R) $ and  $	\mathcal{L}_i(M',\id, \id) \circ \mathcal{L}_i(M,\id, \id)  \leq  	\mathcal{L}_i(M'\circ M,\id, \id)$.
In order to obtain a parametrized lax extension, it suffices to have the following:
\begin{proposition}\label{prop:laxExtParam} We say that $T$ satisfies the \emph{strong interchange property} if for every $f: A \to B$, $g : C \to D$ and $h : X \to Y$, the following square is a weak pullback 
		\begin{center}
		\begin{tikzpicture}[thick,yscale=0.5,xscale=0.7]
		
				\node (A) at (0,1.5) {\scriptsize$T(A,C,X)$};
				\node (B) at (2,0) {\scriptsize$T(A,D,Y)$};
				\node (C) at (4,1.5) {\scriptsize$T(B,D,Y)$};
				\node (D) at (2,3) {\scriptsize$T(B,C,X)$};
				
				\draw [->] (A)  to  node [below left] {\scriptsize$T(\id, g,h)$} (B);
				\draw [->] (C)  to  node [below right] {\scriptsize$T(f, \id, \id)$} (B);
				\draw [->] (D) to  node [above left] {\scriptsize$T(f, \id, \id)$} (A);
				\draw [->] (D)  to node [above right] {\scriptsize$T(\id, g,h)$} (C);
				\node at (2,1.5) {\scriptsize$\mathrm{wpb}$};

		\end{tikzpicture}
	\end{center} 
	If $T$ preserves weak pullbacks in its last two arguments as in Lemma~\ref{lem:WeakPbkLastTwo} and satisfies the strong interchange property, then all the operators $\mathcal{L}_1$, $\mathcal{L}_2$, $\mathcal{L}_3$ and $\mathcal{L}_4$ become equal and are parametrized lax extensions. 
\end{proposition}
Intuitively, the strong interchange property means that the contravariant variable is independent from the covariant variables. For example, if $T(A,C,X)= F(A) \times G(C,X)$ for some functors $F : \Set^\op \to \Set$ and $G: \Set \times \Set \to \Set$, then $T$ satisfies this property. An example which does not satisfy this property is the underlying functor of the state monad $T(A,C,X)= (C \times X)^A$ for which the operators $\mathcal{L}_1$, $\mathcal{L}_2$, $\mathcal{L}_3$ and $\mathcal{L}_4$ are all distinct, and only $\mathcal{L}_4$ is a parametrized lax extension. The functor $T$ satisfies a weaker form of interchange:
\begin{proposition}\label{prop:laxExtParamWeak} We say that $T$ satisfies the \emph{weak interchange property} if for every spans $(g : P \to C, h : P \to D)$ and $(u : Q \to X, v : Q \to Y)$ and split epi $f : A \to B$, the canonical inclusion 
	\[\begin{aligned}
	&(T(\id, g, u ))^* ; (T(f,\id, \id))_*; (T(f,\id, \id))^*; (T(\id,h, v))_*\leq \\
	&(T(f,\id, \id))_*; (T(\id, g, u ))^* ; (T(\id,h, v))_*(T(f,\id, \id))^*
	\end{aligned}\]
is an equality. 
	If $T$ preserves weak pullbacks in its last two arguments as in Lemma~\ref{lem:WeakPbkLastTwo} and satisfies the weak interchange property, then $\mathcal{L}_4$ is a parametrized lax extension.
\end{proposition}
We spelled out the results for functors of type $ \Set^\op \times \Set \times \Set \to \Set$ because they correspond to the underlying functor of parametrized monads, but the constructions presented above extend to arbitrary mixed variance functor $\Set^{v_1} \times \dots \Set^{v_n} \to \Set$ where $\Set^{v_i}$ is either $\Set$ or $\Set^\op$. For example, for a functor $T : \Set^\op \times \Set \times \Set^\op \times \Set \times \Set \to \Set$, we define $\mathcal{L}_4(M,N,O,P,Q)$ as
\[
 (T(\pi_1, \id, \pi_1, \id, \id))_* ;(T(\id, \pi_1, \id, \pi_1,\pi_1))^* ;(T(\id, \pi_2, \id, \pi_2, \pi_2))_* ; (T(\pi_2, \id, \pi_2, \id, \id))^*
\]
and we obtain analogous results.
\subsection{Extending the monad structure}
Recall that for an ordinary natural transformation $\alpha : F \Rightarrow G : \Set \to \Set$ with lax extensions $\Barr{F}$ and $\Barr{G}$ respectively, its canonical extension $\Barr{\alpha}: \Barr{F}\Rightarrow \Barr{G}$ has components $(\Barr{\alpha})_X := (\alpha_X)_* : F(X) \profto G(X)$. If $\Barr{F}$ and $\Barr{G}$ are the Barr extensions of $F$ and $G$, then $\Barr{\alpha}$ is in general an oplax transformation: \ie for a relation $R : X \profto Y$, $\Barr{F}(R) ; (\alpha_Y)_*\leq (\alpha_X)_*; \Barr{G}(R)$. To obtain (strict) naturality, we need $\alpha$ to be a weakly cartesian transformation, \ie that all its naturality squares are weak pullbacks.

\begin{definition}\label{def:laxextensionmonad}
	 A \emph{lax extension} of a strong monad $(T, \mu, \eta, \strength)$ on $\Set$ to $\Rel$ consists a lax extension $\Barr{T}$ of $T$, as in Definition~\ref{def:laxextension}, satisfying the following axioms:
	\begin{enumerate}
		\item for all $ R : X \profto Y$, $(\eta_Y)_*  \circ R\leq  \Barr{T}(R) \circ (\eta_X)_*$;
		\item for all $ R : X \profto Y$, $(\mu_Y)_*  \circ \Barr{T}(\Barr{T}(R)) \leq  \Barr{T}(R) \circ (\mu_X)_*$.
		\item for all $R : A \profto B$ and $S : X \profto Y$, $(\strength_{B,Y})_* \circ (R \times \Barr{T}(S)) \leq \Barr{T}(R \times S) \circ \strength_{A,X}$.
	\end{enumerate}
	If $\Barr{T}$ is a strict extension for $T$ and conditions $1.-3.$ are equalities, we say that $(\Barr{T}, \mu_*, \eta_*, \strength_*)$ is a \emph{strict extension} for $(T, \mu, \eta, \strength)$. 
\end{definition}

Notable examples of monads that do not have a strict extension are the powerset monad, the probability distribution monad, the ultrafilter monad and the global state monad. In the parametrized case, the unit and multiplication are transformations which are both natural in some variables and dinatural in others. We therefore need to understand what it means to extend a dinatural transformation. Recall that for ordinary $1$-categorical functors $F, G : \bicat{C}^{\op}\times \bicat{C} \to \bicat{D}$, a dinatural transformation $F \Rightarrow G$ is a family of morphisms $\theta_X : F(X,X) \to G(X,X)$ in $\bicat{D}$ such that for all $f:X \to Y$ in $\bicat{C}$, the hexagon axiom holds: 
\[G(\id, f)\circ \theta_X \circ F(f,\id) = G(f,\id) \circ \theta_Y \circ F(\id, f)\].
In the general $2$-categorical setting, the notion of $2$-dinaturality takes into account the various possible dualities and we have transformations between functors of type $ \bicat{C}^{\co\op}\times\bicat{C}^{\op}\times \bicat{C}^{\co}\times \bicat{C} \to \bicat{D}$ which satisfy coherences with $2$-morphisms~\cite{2DinatShulman}. In our case, we are only working with posetal $2$-categories so the situation is simplified. For $2$-functors $F, G : \Rel^\co \times \Rel \to \Rel$, a dinatural transformation $F \Rightarrow G$ is simply a family of relations $\theta_X : F(X,X) \profto G(X,X)$ such that for all $R : X\profto Y$, $ \theta_X ; G(R,R)= F(R,R) ;\theta_Y$. It is oplax if $  F(R,R) ;\theta_Y \leq \theta_X ; G(R,R)$ and lax if we have the reverse inclusion. 

\begin{definition}\label{def:LaxParamExtensionMonad}
	For a parametrized monad  $(T ,\eta,\mu)$, a \emph{parametrized lax extension} consists of a lax extensions $\Barr{T}$ as in Definition~\ref{def:ParamLaxExtensionMonadSet}, such that the relations $(\eta_{A,X})_*$ and $	(\mu_{A_1,A_2,A_3,X})_*$ verify:
	\begin{enumerate}
		\item for all $M : A \profto B$ and $R : X \profto Y$,
		\[
		R ;	(\eta_{B, Y})_* \leq (\eta_{A, X})_*; \Barr{T} (M,M,R)
		\]
		\item for all $M: A_1 \profto B_1$, $N: A_2 \profto B_2$, $O: A_3 \profto B_3$ and $R : X \profto Y$,
		\[
		\Barr{T}(M,N, \Barr{T}(N, O, R)) ;(\mu_{B_1, B_2, B_3,Y})_*  \leq (\mu_{A_1, A_2,A_3,X})_*;\Barr{T}(M,O, R)
		\]
	\end{enumerate}
	If the parametrized monad is equipped with a strength $\strength$, then a \emph{strong lax extension for $T$} is a lax extension as above satisfying the additional axiom:
	\begin{enumerate}
		\setcounter{enumi}{2}
		\item for all $M: A_1 \profto B_1$, $N: A_2 \profto B_2$, $R : X_1 \profto Y_1$ and $S : X_2 \profto Y_2$
		\[
		(R\times \Barr{T}(M,N, S) ); (\strength_{A_2, B_2, X_2, Y_2})_* \leq (\strength_{A_1, B_1, X_1, Y_1})_*;  \Barr{T}(M,N, R \times S).
		\]
	\end{enumerate}
	A \emph{strict parametrized extension} consists of a strict extension $\Barr{T}$ as in Definition~\ref{def:strictParamExtension} for which the inequalities $1.-3.$ above are equalities. 
\end{definition}

\begin{proposition}\label{prop:LaxParamExtensionMonad}
	For a strong parametrized monad  $(T ,\eta,\mu, \strength)$, if $\mathcal{L}_4$ is a parametrized lax extension for $T$, then $(\mathcal{L}_4, \eta_*, \mu_*, \strength_*)$ is a parametrized lax extension for $(T ,\eta,\mu, \strength)$ as in Definition~\ref{def:LaxParamExtensionMonad}.
\end{proposition}
Note that if any of $\mathcal{L}_1,\mathcal{L}_2 $ or $\mathcal{L}_3$ are lax extensions, they must be equal to $\mathcal{L}_4$ by Theorem~\ref{th:paramBarr}.7.

\begin{example}
	The parametrized state monad has underlying functor $T$ as in Example~\ref{ex:paramStateMonad}.1. For sets $A,X$, the unit has components $\eta_{A,X} : X \to (A \times X)^{A}$ given by the function $x \mapsto \lambda a. (a,x)$. For sets $A,B,C$ and $X$, the component of the multiplication $	\mu_{A,B,C,X}	: T(A, B,T(B,C,X)) \to T(A,C,X)$ maps $\varphi  \in (B \times (C \times X)^B)^A$ to $\lambda a. (\pi_2(\varphi(a))(\pi_1(\varphi(a)))$. It is a strong monad: for sets $A,B,X$ and $Y$, the component of the strength $\strength_{A,B,X,Y}: X \times (B \times Y)^A \to  (B \times X\times Y)^A $ maps $(x, \psi) \in X \times (B \times Y)^A $ to $ \lambda a. (\pi_1(\psi(a)), x, \pi_2(\psi(a)))$. We obtain a parametrized lax extension $(\mathcal{L}_4, \eta_*, \mu_*, \strength_*)$, it corresponds to the monad structure induced from the parametrized adjunction for the (double) cartesian closed structure of $\Rel$. 
\end{example}

\section{Parametrized behavioral equivalence}\label{sec:bisim}

Many coinductive preorders and equivalences to reason on systems can be described coalgebraically. In this paper, we will focus on the three following notions: Aczel-Mendler bisimulation, $\Barr{F}$-bisimulation (for some lax extension $\Barr{F}$) and behavioral equivalence.
An important problem is how to develop sound and/or complete proof principles to establish behavioral equivalence of states, which is often difficult to show directly in practice. This problem has been well-studied for set-functors~\cite{aczel2005final}, and it has also been addressed in the quantale-enriched setting~\cite{worrell2000coalgebras, double}.  

 \subsection{Non-parametrized bisimulation and behavioral equivalence}\label{subsec:nonParamBisimExtension}
 For an endofunctor $F:\Set \to \Set$ and coalgebras $(X, \alpha: X \to F(X))$, $(Y, \beta : Y \to F(Y))$, a relation  $R : X \profto Y$  is a behavioral equivalence if there exists a coalgebra $(Z, \zeta:Z \to F(Z))$ and coalgebra morphisms $f : (X,\alpha) \to (Z, \zeta)$, $g : (Y, \beta) \to (Z, \zeta)$ as on the left below such that the diagram on the right is a weak pullback:

 	\begin{center}
 	\begin{tikzpicture}[thick, xscale=0.5, yscale=0.5]
 		\node (A) at (0,0) {\scriptsize$F(X)$};
 		\node (B) at (3,-1) {\scriptsize$F(Z)$};
 		\node (C) at (0,2) {\scriptsize$X$};
 		\node (D) at (3,1) {\scriptsize$Z$};
 		\node (E) at (6,2) {\scriptsize$Y$};
 		\node (F) at (6,0) {\scriptsize$F(Y)$};
 		
 		\draw [->] (C) -- node [left] {\scriptsize$\alpha$} (A);
 		\draw [->] (A) -- node [below left] {\scriptsize$F(f)$} (B);
 		\draw [->] (F) -- node [below right] {\scriptsize$F(g)$} (B);
 		\draw [->] (C) -- node [above] {\scriptsize$f$} (D);
 		\draw [->] (E) -- node [above] {\scriptsize$g$} (D);
 		\draw [->] (E) -- node [right] {\scriptsize$\beta$} (F);
 		\draw [->] (D) -- node [right] {\scriptsize$\zeta$} (B);
 			\begin{scope}[xshift=8.5cm]
 			\node (A) at (0,.5) {\scriptsize$X$};
 			\node (B) at (2,-1) {\scriptsize$Z$};
 			\node (C) at (4,0.5) {\scriptsize$Y$};
 			\node (D) at (2,2) {\scriptsize$\graph{R}$};
 			
 			\draw [->] (A)  to  node [below left] {\scriptsize$f$} (B);
 			\draw [->] (C)  to  node [below right] {\scriptsize$g$} (B);
 			\draw [->] (D) to  node [above left] {\scriptsize$\pi_1^R$} (A);
 			\draw [->] (D)  to node [above right] {\scriptsize$\pi_2^R$} (C);
 			\node at (2,.5) {\scriptsize$\mathrm{wpb}$};
 		\end{scope}
 	\end{tikzpicture}
 \end{center}

 	Two states $x\in X$ and $y\in Y$ are \emph{behaviorally equivalent} if there exists a relation $R$ as above such that $(x,y) \in R$.
 An important characterization of behavioral equivalence is that if $F$ has a final coalgebra $\out_F : \nu F \to F(\nu F)$, then two states $x$ and $y$ are behaviorally equivalent if and only if they have the same \emph{final semantics} \ie they are mapped to equal states in the final coalgebra.  We typically restrict to the case where $\alpha = \beta$ and $R$ is an actual equivalence relation, not just a difunctional relation as above.
 
 Two states $x \in X$ and $y \in Y$ are \emph{Aczel-Mendler bisimilar} (AM-bisimilar) if there exists a relation $R : X \profto Y$ and a coalgebra $\gamma: \graph{R} \to F(\graph{R})$ such that $(x,y) \in R$ and the projection maps $\pi_1 : \graph{R} \to X$ and $\pi_2: \graph{R} \to Y$ are coalgebra morphisms.

 \begin{definition}
 	For coalgebras $(X, \alpha)$ and $(Y, \beta)$ as above, a relation $R : X \profto Y$ is \emph{sound} for behavioral equivalence if for any $(x,y) \in R$, the states $x$ and $y$ are behaviorally equivalent. It is \emph{complete} if the reverse inclusion holds and it is \emph{fully abstract} if both inclusions hold.
 \end{definition}
 
 \begin{theorem}[\cite{rutten2000universal}]
 	 AM-bisimilarity is sound for behavioral equivalence; if $F$ preserves weak-pullbacks, then it is also fully abstract.
 \end{theorem}
 
 We are mainly interested in this paper in the notion of equivalence induced by a lax extension $\bar{F}$ of $F$, two states $x \in X$ and $y \in Y$ are \emph{$\bar{F}$-bisimilar} if there exists a relation $R : X \profto Y$ such that $(x,y) \in R$ and $R$ is included in $ \alpha_*;\bar{F}(R) ; 
 \beta^* $.
 
Bisimulations induced by lax extensions subsume AM-bisimulations. 
  \begin{proposition}\label{prop:AMsoundforLax}
	If $R$ is an AM-bisimulation, then it is a $\bar{F}$-bisimulation for any lax extension $\bar{F}$.
\end{proposition}

\begin{theorem}[\cite{marti2015lax,goncharov2025relators}]\label{thm:LaxExtFullyAbstract}
	For a lax extension $\bar{F}$, $\bar{F}$-bisimilarity is complete for behavioral equivalence and if $\bar{F}$ is normal, then $\bar{F}$-bisimilarity is fully abstract.
\end{theorem}

\subsection{Parametrized bisimulation from lax extensions}
We fix a parametrized lax-extension $\Barr{T}$ of a functor $T : \Set^\op \times \Set \times \Set \to \Set$. We first need to determine what is the notion of coalgebra in this setting. In this paper, we do not follow the approach of~\cite{atkey2009algebras} where a parametrized coalgebra would be a functor $c : \Set^\op  \to \Set$ equipped with a family of functions $\alpha_{A,C}: c(C)\to T(A,C,c(C))$ natural in $A$ and dinatural in $C$. 
This definition requires the state space to depend on the output parameter, which is too restrictive for some of the examples we consider.
Moreover, we want to be able to consider a single coalgebra $\alpha : X \to T(A,C,X)$ at a time, together with relations between parameters chosen by the user, without requiring those relations to be defined for all sets. In our setting, we fix coalgebras $\alpha : X \to T(A,C,X)$ and $\beta : Y \to T(B,D,Y)$ and relations $M : A \profto B$ and $N : C \profto D$ on the parameters.  We want to understand what it means for a relation $R : X \profto Y$ between states to be a bisimulation with respect to $(M,N)$. The most straightforward generalization to our setting is the lax-extension notion of bisimulation:
\begin{definition}\label{def:parambisim}
	A relation $R : X \profto Y$ is a $\Barr{T}$-bisimulation w.r.t $(M,N)$ if it satisfies $R \leq  \alpha_* ; \Barr{T}(M,N, R) ; \beta^*$.
\end{definition}
Since $\Barr{T}$ is monotone in its third argument, the mapping $R \mapsto \alpha_* ; \Barr{T}(M,N, R) ; \beta^*$ has a greatest fixpoint which we call \emph{$\Barr{T}$-bisimilarity w.r.t $(M,N)$}. Note that using the formalism of~\cite{nora2025relational}, $\Barr{T}(M,N, -)$ is an example of \emph{relational connector} between the functors $T(A,C,-)$ and $T(B,D,-)$ and the notion of $\Barr{T}$-bisimilarity w.r.t $(M,N)$ coincides with the notion of bisimilarity induced by this relational connector.

Assume now that $\beta =\alpha$ (so that $M : A\profto A$ and $N : C \profto C$), in the non-parametrized case, the axioms of lax extension ensure that $\Barr{T}$-bisimulations are closed under composition and identity which implies that $\Barr{T}$-bisimilarity is a preorder. In our case, we need additional conditions on $M$ and $N$:
\begin{lemma}\label{lem:paramcond}
	If $M \leq \id_A$ and $ \id_C \leq N$, then $\id_X$ is a $\Barr{T}$-bisimulation w.r.t $(M,N)$. If $M \leq M \circ M $, $ N \circ N \leq N $ and $R$ is a $\Barr{T}$-bisimulation w.r.t $(M,N)$, then so is $R \circ R$.
\end{lemma}
In practice, the conditions on $M$ are too restrictive and we typically want $M$ to be a preorder. While $\Barr{T}$-bisimilarity will not be a preorder for arbitrary coalgebras $\alpha$, it will hold for the specific cases we are interested in.

To obtain that $\Barr{T}$-bisimilarity is an equivalence relation, we need it to be symmetric and it suffices to require that $\Barr{T}$ commutes with the operation of taking the transpose of relations:

\begin{lemma}
	Assume that $M^t = M$, $N^t = N$ and $\Barr{T}$ is symmetric (axiom $4.$ in Definition~\ref{def:ParamLaxExtensionMonadSet}). If $R$ is a $\Barr{T}$-bisimulation w.r.t $(M,N)$, then so is $R^t$.
\end{lemma}
\subsection{Parametrized AM-bisimulation}\label{sec:AMbisim}

A relation $R : X \profto Y$ is an \emph{AM-bisimulation w.r.t $(M,N)$} if there exists a coalgebra $\gamma : \graph{R} \to T(\graph{M}, \graph{N}, \graph{R})$ such that the two pentagons below commute:
	\begin{center}
	\begin{tikzpicture}[thick, xscale=1.2, yscale=0.8]
		\node (A) at (0,0) {\scriptsize$T(A,C,X)$};
		\node (A1) at (1.5,-1.6) {\scriptsize$T(\graph{M},C,X)$};
		\node (B) at (3,0) {\scriptsize$T(\graph{M}, \graph{N}, \graph{R})$};
		\node (C) at (0,1.4) {\scriptsize$X$};
		\node (D) at (3,1.4) {\scriptsize$\graph{R}$};
		\node (E) at (6,1.4) {\scriptsize$Y$};
		\node (F) at (6,0) {\scriptsize$T(B,D,Y)$};
		\node (F1) at (4.5,-1.6) {\scriptsize$T(\graph{M},D,Y)$};
		
		\draw [->] (C) -- node [left] {\scriptsize$\alpha$} (A);
		\draw [->] (B) -- node [fill=white] {\scriptsize$T(\id, \pi_1^N,\pi_1^R)$} (A1);
		\draw [->] (A) -- node [left] {\scriptsize$T(\pi_1^M, \id, \id )$} (A1);
		\draw [->] (F) -- node [right] {\scriptsize$T(\pi_2^M, \id, \id )$} (F1);
		\draw [->] (B) -- node [fill=white] {\scriptsize$T(\id, \pi_2^N,\pi_2^R)$} (F1);
		\draw [->] (D) -- node [above] {\scriptsize$\pi_1^R$} (C);
		\draw [->] (D) -- node [above] {\scriptsize$ \pi_2^R$} (E);
		\draw [->] (E) -- node [right] {\scriptsize$\beta$} (F);
		\draw [->] (D) -- node [right] {\scriptsize$\gamma$} (B);
	\end{tikzpicture}
\end{center}
This notion can be defined without any restrictions on the parameter relations $M$ and $N$, and from the axioms of parametrized lax extension, we obtain soundness similarly to Proposition~\ref{prop:AMsoundforLax}:

\begin{lemma}\label{lem:AMLaxExtSoundness}
	If $R$ is an AM-bisimulation w.r.t $(M,N)$, then it is a $\Barr{T}$-bisimulation w.r.t $(M,N)$.
\end{lemma}

\subsection{Parametrized behavioral equivalence}\label{subsec:paramBehequiv}
We consider the case where $\alpha = \beta$ and $M : A\profto A$ and $N : C \profto C$ are endorelations. For the notion of $(M,N)$-behavioral equivalence, we also require the relation $N$ to be difunctional: there exists $h : C \to W$ and $k : C \to W$ such that the left square below is a weak-pullback
	\begin{center}
	\begin{tikzpicture}[thick,yscale=0.6, xscale=0.5]
		\node (A) at (0,1.5) {\scriptsize$C$};
	\node (B) at (2,0) {\scriptsize$W$};
	\node (C) at (4,1.5) {\scriptsize$C$};
	\node (D) at (2,3) {\scriptsize$\graph{N}$};
	
	\draw [->] (A)  to  node [below left] {\scriptsize$h$} (B);
	\draw [->] (C)  to  node [below right] {\scriptsize$k$} (B);
	\draw [->] (D) to  node [above left] {\scriptsize$\pi_1^N$} (A);
	\draw [->] (D)  to node [above right] {\scriptsize$\pi_2^N$} (C);
	\node at (2,1.5) {\scriptsize$\mathrm{wpb}$};
	
	\begin{scope}[xshift=7.5cm]
			\node (A) at (0,1.5) {\scriptsize$X$};
		\node (B) at (2,0) {\scriptsize$Z$};
		\node (C) at (4,1.5) {\scriptsize$X$};
		\node (D) at (2,3) {\scriptsize$\graph{R}$};
		
		\draw [->] (A)  to  node [below left] {\scriptsize$f$} (B);
		\draw [->] (C)  to  node [below right] {\scriptsize$g$} (B);
		\draw [->] (D) to  node [above left] {\scriptsize$\pi_1^R$} (A);
		\draw [->] (D)  to node [above right] {\scriptsize$\pi_2^R$} (C);
		\node at (2,1.5) {\scriptsize$\mathrm{wpb}$};
	\end{scope}
	
	\begin{scope}[xshift=15cm,yshift=1cm,xscale=1.6]
			\node (A) at (0,0) {\scriptsize$T(A,C,X)$};
		\node (B) at (3,-1) {\scriptsize$T(\graph{M},W,Z)$};
		\node (C) at (0,2) {\scriptsize$X$};
		\node (D) at (3,1) {\scriptsize$Z$};
		\node (E) at (6,2) {\scriptsize$X$};
		\node (F) at (6,0) {\scriptsize$T(A,C,X)$};
		
		\draw [->] (C) -- node [left] {\scriptsize$\alpha$} (A);
		\draw [->] (A) -- node [below left] {\scriptsize$T(\pi_1^M, h,f)$} (B);
		\draw [->] (F) -- node [below right] {\scriptsize$T(\pi_2^M, k, g)$} (B);
		\draw [->] (C) -- node [above] {\scriptsize$f$} (D);
		\draw [->] (E) -- node [above] {\scriptsize$g$} (D);
		\draw [->] (E) -- node [right] {\scriptsize$\alpha$} (F);
		\draw [->] (D) -- node [right] {\scriptsize$\zeta$} (B);
	\end{scope}
	\end{tikzpicture}
\end{center}
In this setting, two states $x,y \in X$ are \emph{$(M,N)$-behaviorally equivalent} if there exist a relation $R : X\profto X$ and functions $f : X \to Z$, $g : X \to Z$ such that the middle square above is a weak-pullback, $(x,y)\in R$, and there exists a coalgebra $\zeta : Z \to  T(\graph{M},W,Z)$ such that the two squares on the right commute.
\begin{lemma}\label{lem:paramBehEquivFinal}
	For $\alpha$ and $(M,N)$ as above, if the functor $T(\graph{M}, W, -) : \Set \to \Set$ has a final coalgebra $\omega : \Omega \to T(\graph{M}, W, \Omega)$, then two states $(x,y) \in X \times X$ are behaviorally equivalent if and only if $\sem{x}_1 = \sem{y}_2$ where $\sem{-}_1$ and $\sem{-}_2$ are the two unique coalgebra morphisms as below:
	\begin{center}
	\begin{tikzpicture}[thick, xscale=0.9, yscale=0.5]
		\node (A) at (0,0) {\scriptsize$T(A,C,X)$};
		\node (A1) at (0,-2) {\scriptsize$T(\graph{M},W,X)$};
		\node (B) at (3,-3) {\scriptsize$T(\graph{M},W,\Omega)$};
		\node (C) at (0,2) {\scriptsize$X$};
		\node (D) at (3,1) {\scriptsize$\Omega$};
		\node (E) at (6,2) {\scriptsize$X$};
		\node (F) at (6,0) {\scriptsize$T(A,C,X)$};
		\node (F1) at (6,-2) {\scriptsize$T(\graph{M},W,X)$};
		
		\draw [->] (C) -- node [left] {\scriptsize$\alpha$} (A);
		\draw [->] (A1) -- node [below left] {\scriptsize$T(\id, \id ,\sem{-}_1)$} (B);
		\draw [->] (A) -- node [left] {\scriptsize$T(\pi_1^M, h,\id)$} (A1);
		\draw [->] (F) -- node [ right] {\scriptsize$T(\pi_2^M, k, \id)$} (F1);
		\draw [->] (F1) -- node [below right] {\scriptsize$T(\id, \id, \sem{-}_2)$} (B);
		\draw [->] (C) -- node [above] {\scriptsize$ \sem{-}_1$} (D);
		\draw [->] (E) -- node [above] {\scriptsize$ \sem{-}_2$} (D);
		\draw [->] (E) -- node [right] {\scriptsize$\alpha$} (F);
		\draw [->] (D) -- node [right] {\scriptsize$\omega$} (B);
	\end{tikzpicture}
\end{center}
\end{lemma}
\begin{proposition}\label{prop:AMBehEquivsoundness}
	If $T$ preserves weak-pullbacks in the last two variables, then $(M,N)$-behavioral equivalence is sound for AM-bisimilarity w.r.t $(M,N)$.
\end{proposition}

Before stating the full-abstraction result analogous to Theorem~\ref{thm:LaxExtFullyAbstract}, we give a useful lemma which is similar to Proposition III.1.4.3~\cite{Hofmann_Seal_Tholen_2014} for the non-parametrized setting:
\begin{lemma}\label{lem:laxdblefunctor}
	For relations $M,N,R$ and functions $f,g,h,k,i,j$ as below
	\begin{center}
		\begin{tikzpicture}[thick,xscale=0.4, yscale =0.4, decoration={ markings, mark=at position 0.5 with {\arrow{|}}}]
			\node (A) at (0,3) {$A_1$};
			\node (B) at (4,3) {$A_2$};
			\node (C) at (0,0) {$A_3$};
			\node (D) at (4,0) {$A_4$};
			
			\draw [->] (A) -- node [left] {$f$} (C);
			\draw [->, postaction=decorate] (A) -- node [above] {$M$} (B);
			\draw [->] (B) -- node [right] {$g$} (D);
			\begin{scope}[xshift=8cm]
				\node (A) at (0,3) {$B_1$};
				\node (B) at (4,3) {$B_2$};
				\node (C) at (0,0) {$B_3$};
				\node (D) at (4,0) {$B_4$};
				
				\draw [->] (A) -- node [left] {$h$} (C);
				\draw [->, postaction=decorate] (C) -- node [below] {$N$} (D);
				\draw [->] (B) -- node [right] {$k$} (D);
				
			\end{scope}
			\begin{scope}[xshift=16cm]
				\node (A) at (0,3) {$X_1$};
				\node (B) at (4,3) {$X_2$};
				\node (C) at (0,0) {$X_3$};
				\node (D) at (4,0) {$X_4$};
				
				\draw [->] (A) -- node [left] {$i$} (C);
				\draw [->, postaction=decorate] (C) -- node [below] {$R$} (D);
				\draw [->] (B) -- node [right] {$j$} (D);
			\end{scope}
		\end{tikzpicture}
		
	\end{center}
	if $\Barr{T}$ is a parametrized lax extension, then 
	\[
	\Barr{T} (f^*;M;g_*, h_*;N;k^*, i_*;R;j^*) = (T(f,h,i))_*; \Barr{T} (M, N, R) ;  (T(g,k,j))^*
	\]	
\end{lemma}

\begin{theorem}\label{th:BehEquivExtFullabstract}
	$\Barr{T}$-bisimilarity w.r.t $(M,N)$ is complete for $(M,N)$-behavioral equivalence and if $\Barr{T}$ is normal, then it is fully abstract.
\end{theorem}

\section{Case studies}\label{sec:calculus}

\subsection{Mealy Machines}
We first consider a simple example of deterministic finite state machines: for fixed sets $I$ and $O$ of inputs and outputs respectively, a \emph{Mealy machine} is a triple $M=(X\in \Set, \nextt : X \times I \to X,\out : X\times I \to O)$, or equivalently a coalgebra for the functor $F: X \mapsto (X \times O)^I$~(\eg~\cite{pattinson2003introduction}) which has a strict extension $\Barr{F}$ to $\Rel$. For two machines $M_1$ and $M_2$, a relation $R: X_1 \profto X_2$ between their state spaces is an $\Barr{F}$-bisimulation if for every $(x_1, x_2) \in R$, and $i \in I$, $(\nextt_1(x_1,i),\nextt_2(x_2,i)) \in R$ and $\out_1(x_1, i) = \out_2(x_2,i)$.

In the parametrized case, we allow the sets of inputs and outputs to vary and consider the parametrized functor $T: (I, O, X) \mapsto (O \times X)^I$ which verifies the conditions of Proposition~\ref{prop:laxExtParamWeak} and therefore has a parametrized lax extension $\Barr{T}$. We can now consider two machines $M_1, M_2$ on different sets of inputs $I_1, I_2$ and outputs $O_1,O_2$. For relations $M : I_1 \profto I_2$ and $N: O_1 \profto O_2$ specifying how inputs and outputs can be related, the induced notion of $\Barr{T}$-bisimulation w.r.t $(M,N)$ relaxes the notion of bisimulation above by requiring only that on two related inputs, the two corresponding outputs should be related: for all $(x_1, x_2) \in R$, $(i_1, i_2)\in M$,  $(\nextt_1(x_1,i_1),\nextt_2(x_2,i_2)) \in R$ and $(\out_1(x_1, i_1) ,\out_2(x_2,i_2)) \in N$.

\subsection{Parametrized state contextual equivalence}\label{sect:parstate}

We finally instantiate our framework to a higher-order calculus with parametrized monadic effects and sketch how it supports notions of parametrized contextual equivalence that cannot be directly addressed using the ordinary notion of lax extension (Definition~\ref{def:laxextension}). 

We fix a (possibly partial) monoid $(\mathcal{L}, \varepsilon, \cdot, \leq)$. Our motivating example is the monoid of memory regions for a fixed set of locations $\Loc$, a memory region is simply a subset $L \subseteq \Loc$. The monoid multiplication $L_1 \cdot L_2$ is defined when the two subsets are disjoint and it is given by their union. We consider a monadic type system with state and computation types, programs are terms in a standard $\lambda$-calculus
given in fine-grain call-by-value style, which we give below.
{\small
	\begin{align*}
\textbf{State types }L_1, L_2&::=  \varepsilon \mid L_1 \cdot L_2\\
\textbf{Computation types }A_1, A_2&::= \bool \mid \unit  \mid (A_1, L_1) \to (A_2, L_2)
\end{align*}}
	\begin{center} 
		\vspace*{-8mm}  
		 \small
		\begin{mathpar}\mprset{sep=0.4em,  vskip =0.1ex}
			\inferrule* []{\\}{ \Gamma , x: A \vdashv[] x: A}
			\and 
			\inferrule* []{\\}{ \Gamma \vdashv[] \true : \bool}	
			\and
			\inferrule* []{\\}{\Gamma \vdashv[] \false :\bool}	
			\and
			\inferrule* []{\\}{\Gamma \vdashc[] \star :\unit}	
			\and
			\inferrule* []{\Gamma, x: A_1 ; L_1 \vdashc[] M : A_2 ; L_2}{\Gamma \vdashv[] 
			\lambda(x:A_1;L_1). M: (A_1, L_1) \to (A_2, L_2)}	
			\and
			\inferrule* []{\Gamma  \vdashv[] V: A}{\Gamma;L  \vdashc[] \return_L V: A;L}	
			\and
			\inferrule* []{\Gamma  \vdashv[] V: (A_1, L_1) \to (A_2, L_2)\\ \Gamma  \vdashv[] W:A_1}{\Gamma; L_1  \vdashc[] VW : A_2 ; L_2 }
			\and
			\inferrule* []{ \Gamma ; L_1 \vdashc[] M : A ; L_2 \\ \Gamma, x: A ; L_2  \vdashc[] N : B ; L_3}
			{\Gamma ; L_1\vdashc[] \letin{ x = M}{N}: B ; L_3}
			\and
			\inferrule* []{\Gamma  \vdashv[] V_i: A_1 \\ \op : (A_1,L_1)^n \to (A_2, L_2)}{\Gamma ; L_1 \vdashc[] \op(V_1, \dots, V_n): A_2 :L_2}	
			\and 
			\inferrule* []{ \Gamma ; L_1 \vdashc[] M : A ; L_2}
			{\Gamma ; L_1 \cdot L_3 \vdashc[]M : A ; L_2\cdot L_3}
		\end{mathpar}
			\vspace{-0.5cm}
	\end{center}
  The metavariable $\op$ stands for an operation accessing or modifying the underlying store. Examples of operations are $\mathtt{set}^L_\ell : (\mathbf{bool}, L) \rightarrow (\mathbf{1}, L)$ and $\mathtt{get}^L_\ell : (\mathbf{1}, L) \rightarrow (\mathbf{bool}, L)$ where $\ell$ is a location in the region $L$ and $\mathtt{new}^L_\ell :\; (\mathbf{bool}, L) \rightarrow (\mathbf{1}, L \cdot \{l\})$ for a new location $\ell \in \Loc \setminus L$. 
For simplicity, we restrict locations to carry only boolean values. 

We denote by $\Lambda(L_1, L_2, A)$ the
set of terms $M$ such that $\emptyset; L_1 \vdashc[] M : A ; L_2$ and by $ \Val(A)$ is the set of values $V$ such that
$\emptyset\vdashv[] V:A$.
The small step operational semantics is an indexed family of partial functions $\{\Lambda( L_1, L_2,A) \rightharpoonup T(L_1, L_2,\Lambda(L_1, L_2,A) \}_{L_1,L_2,A}$ where $T$ is the parameterized state monad defined as
$T(L_1,L_2,X)=(\sem{L_2}\times X)^{\sem{L_1}}$ and 
$\sem{L}$ is the set of stores (functions mapping locations in $L$
to boolean values). Examples of small step reductions are:
\[
(\return_L (\lambda (x:A;L_1). M)V, s) \to (M[V/x], s) \text{ and } (\mathtt{new}^L_\ell \true, s) \to (\return \star, [s, \ell \mapsto \true])
\]
where $[s, \ell \mapsto \true] : L \cup \{\ell\} \to \{\true, \false\}$ maps $\ell' \in L$ to $s(\ell')$ and $\ell$ to $\true$.
The induced big step operational semantics is an indexed family of functions $\{\Downarrow: \Lambda(L_1, L_2, A) \to T( L_1, L_2, \Val(A)) \}_{A,L_1,L_2}$.

We are interested in comparing terms, 
a term relation consists of a pair of families of relations
$(R^{\Val} : \Val \profto \Val, R^\Lambda: \Lambda \profto \Lambda)$ 
indexed following the aforementioned pattern. 
Standard contextual equivalence is the largest equivalence that is adequate 
with respect to the big step operational semantics and compatible with respect to 
all term constructors. The theory we have developed allows us to refine this notion by modulating the power of contexts depending on which memory region they have access to. This is precisely what happens in the 
so-called adversarial rules often used in security proofs (e.g. ~\cite{ABGGKS21}), where the role of the context 
is played by an adversary that has access only to some (and not all) of the memory locations available 
to the honest parties.

Consider a memory region $L = L_1 \uplus L_2$ and a term $\emptyset; L_1 \vdashc[] M : A ; L_1$ which has access only to the region $L_1$. Assume that we make changes to the memory outside of $L_1$, for example,  define $N := \letin{ x = \mathtt{set}^L_{\ell_2}\true}{M}$ where $\ell_2 \in L_2$,  then a context which has access to the whole region $L$ will be able to distinguish $M$ from $N$ (even if the variable $x$ does not occur in $M$). Indeed, for $M$ and $N$ to be contextually equivalent, the following must hold
\[
\forall s, r \in \sem{L}, (C[M] (s)\Downarrow (\true, r) \quad \Leftrightarrow \quad C[N] (s)\Downarrow (\true, r)) 
\]
 for every context $C[\cdot]$ such that $\emptyset; L \vdashc[] C[M],C[N]: \bool ; L$. We can easily construct a context in which the equivalence above fails, for instance by using $\mathtt{get}^L_{l_2}$.

Instead, $M$ and $N$ are equivalent for weaker contexts that can access only the region $L_1$.
We define an equivalence relation $E$ on stores $\sem{L} = \Set(L, \{\true, \false\})$ as 
\[E:= \{ (s,s') \mid \forall \ell \in L_1, s(\ell) = s'(\ell)\}\]
where two stores are equivalent when they are equal in the restricted region $L_1$.
We say that two terms $M, N$ are $L_1$-equivalent if for all $(s_1, s_1') \in E$:
\[\begin{aligned}
	(\exists s_2 \in \sem{L}, C[M] (s_1)\Downarrow (\true, s_2))& \Rightarrow (\exists s_2' \in \sem{L}, C[N] (s_1')\Downarrow (\true, s_2') \wedge (s_2, s_2') \in E)\\
	(\exists s_2' \in \sem{L}, C[N] (s_1')\Downarrow (\true, s_2')) &\Rightarrow  (\exists s_2 \in \sem{L}, C[M] (s_1)\Downarrow (\true, s_2)  \wedge (s_2, s_2') \in E)
\end{aligned}
\]
for every context $C[\cdot]$ containing only constructors with locations in $L_1$ and such that $\emptyset; L \vdashc[] C[M],C[N]: \bool ; L$. 
Coming back to our example, $M$ and $N$ will be equivalent for this notion of refined equivalence. This equivalence can further be used to capture a finer form of commutativity for effects acting on disjoint sets of memory locations which we cannot show using ordinary contextual equivalence since the state monad is a typical example of a non-commutative monad where the effect of writing on the memory is not commutative. We can also recover standard contextual equivalence by replacing $E$ with the equality relation and allowing contexts $C[\cdot]$ to be built with constructor accessing and modifying all locations.

This refined notion of contextual equivalence can be coinductively characterized as the largest term equivalence relation $R= (R^{\Val}, R^\Lambda)$ satisfying the following conditions:
\begin{enumerate}
	\item $R$ is $(E,E)$-adequate with respect to the big step operational semantics, which corresponds to a $\Barr{T}$-bisimulation w.r.t $(E,E)$: $R^\Lambda(L,L, A) \leq (\Downarrow)_* ; \Barr{T}(E,E, R^{\Val}(A));  (\Downarrow)^*$ where $\Barr{T}$ is the lax-extension of Example~\ref{ex:statenostrictextension}.
	\item $R$ is $L_1$-compatible: when defining the compatible closure of $R$, we restrict the standard definition to term constructors with locations in $L_1$. 
\end{enumerate}
So far, we used the same relation $E$ for both parameters to restrict the scope of testing contexts but our framework also allows the user to specify different predicates on the memory region before and after the execution of a computation.

\paragraph*{Conclusion and future work}

In this paper, we have extended the theory of lax extensions to parametrized functors and monads on sets, and showed that it allows to refine the usual notions of behavioral equivalence by studying conditions on the parameters to derive sound and complete proof principles establishing behavioral equivalence. The parametrized setting is particularly well-suited for monadic effects where the parameters correspond to a notion of pre and post condition updated by the execution of a program such as global state and I/O (input/output) effects.

In the ordinary non-parametrized setting, there is a close connection between predicate liftings and lax extensions~\cite{baldan2014behavioral,marti2015lax,wild2022characteristic,kupferman_kantorovich_2023,goncharovstacs25} which we aim to develop in the parametrized case. In particular, the adequate notion of coalgebraic logics derived from liftings and lax extensions in the parametrized case remains to be investigated.

\newpage
\bibliography{biblio}
\appendix

\section{Proofs for Section~\ref{sec:Barr}}\label{sec:appendixlaxextensions}
\subsection{Extending parametrized functors from $\Set$ to $\Rel$}

\begin{lemma}
If a function $f: A \to B$ is split epi, then $f^*; f_* =\id_B: B \profto B$ and if a function $f: A \to B$ is split mono, then $f_*: f^*=\id_A : A \profto A$.
\end{lemma}

\begin{lemma}
	For relations $R : A \profto  B$ and $S : B\profto C$, if we take the pullback as on the left below (or equivalently the composition of the two tabulation spans), then the unique morphism $u: P \to \graph{S \circ R}$ as on the right is split epi.
	
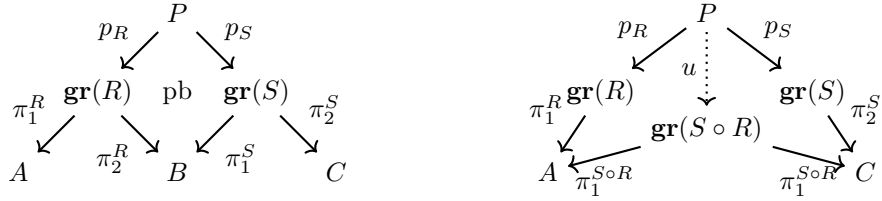
\begin{center}
		\begin{tikzpicture}[thick,xscale=0.7, yscale=0.7]
			\node (A) at (2.5,0) {$A$};
			\node (Y1) at (4,1.5) {$\graph{R}$};
			\node (B) at (5.5,0) {$B$};
			\draw [->] (Y1)  to node [above left] {$\pi_1^R$} (A);
			\draw [->] (Y1)  to node [below left] {$\pi_2^R$} (B);
			
			\node (X2) at (7,1.5) {$\graph{S}$};
			\node (P) at (5.5,3) {$P$};
			\node (C) at (8.5,0) {$C$};
			\draw [->] (P) to node [above left] {$p_R$} (Y1);
			\draw [->] (P) to node [above right] {$p_S$} (X2);
			\draw [->] (X2)  to node [below right] {$\pi_1^S$} (B);
			\draw [->] (X2)  to node [above right] {$\pi_2^S$} (C);

			\node at (5.5,1.5) {$\mathrm{pb}$};
			
			\begin{scope}[xshift=10cm]
	\node (A) at (2.5,0) {$A$};
\node (Y1) at (3.5,1.5) {$\graph{R}$};
\node (B) at (5.5,0.8) {$\graph{S\circ R}$};
\draw [->] (Y1)  to node [above left] {$\pi_1^R$} (A);
\draw [->] (B)  to node [below] {$\pi_1^{S\circ R}$} (A);

\node (X2) at (7.5,1.5) {$\graph{S}$};
\node (P) at (5.5,3) {$P$};
\node (C) at (8.5,0) {$C$};
\draw [->] (P) to node [above left] {$p_R$} (Y1);
\draw [->, dotted] (P) to node [ left] {$u$} (B);
\draw [->] (P) to node [above right] {$p_S$} (X2);
\draw [->] (B)  to node [below] {$\pi_1^{S\circ R}$} (C);
\draw [->] (X2)  to node [above right] {$\pi_2^S$} (C);

			\end{scope}
		\end{tikzpicture}
	\end{center}

\end{lemma}

\begin{proof}[Proof of Theorem~\ref{th:paramBarr}]
	We will only prove the theorem for $\mathcal{L}_1$, the proofs for $i={2,3,4}$ follow the same pattern.
	\begin{enumerate}
		\item[]
	\item To obtain that $\mathcal{L}_1$ is normal, we use the fact that the tabulation of the identity relation is isomorphic to the identity span.
	\item If $M' \leq M : A \profto B$, $N\leq N': C \profto D$ and $R \leq R' : X \profto Y$, then there exist functions $\alpha: \graph{M'}\to \graph{M}$, $\beta : \graph{N} \to \graph{N'}$ and $\gamma : \graph{R} \to \graph{R}'$ such that $\pi_1^{M} \circ \alpha = \pi_1^{M'}$, $\pi_2^{M} \circ \alpha = \pi_2^{M'}$, $\pi_1^{N'} \circ \beta = \pi_1^N$, $\pi_2^{N'} \circ \beta = \pi_2^N$, $\pi_1^{R'} \circ \gamma = \pi_1^R$ and $\pi_2^{R'} \circ \gamma = \pi_2^R$. We obtain:
	\[	\begin{aligned}
		&\mathcal{L}_1(M,N,R)
		=(T(\id, \pi_1^N, \pi_1^R))^*;
		(T(\pi_1^{M}, \id, \id))_*  ;(T(\pi_2^{M}, \id, \id))^* ;(T(\id, \pi_2^N, \pi_2^R))_* \\
		\leq&(T(\id, \pi_1^N, \pi_1^R))^*;
		(T(\pi_1^{M}, \id, \id))_*  ;(T(\alpha, \id, \id))_* ;\\
		&\qquad\qquad\qquad (T(\alpha, \id, \id))^* ;(T(\pi_2^{M}, \id, \id))^* ;(T(\id, \pi_2^N, \pi_2^R))_* \\
		=&(T(\id, \pi_1^N, \pi_1^R))^*;
		(T(\pi_1^{M'}, \id, \id))_*  ;(T(\pi_2^{M'}, \id, \id))^* ;(T(\id, \pi_2^N, \pi_2^R))_* \\
		=&(T(\id, \pi_1^{N'}, \pi_1^{R'}))^*;(T(\id, \beta, \gamma))^*;
		(T(\pi_1^{M'}, \id, \id))_*  ;\\
		&\qquad\qquad\qquad (T(\pi_2^{M'}, \id, \id))^* ;(T(\id, \beta, \gamma))_*;(T(\id, \pi_2^{N'}, \pi_2^{R'}))_* \\
		\leq & (T(\id, \pi_1^{N'}, \pi_1^{R'}))^*;
		(T(\pi_1^{M'}, \id, \id))_* ;(T(\id, \beta, \gamma))^*;\\
		&\qquad\qquad\qquad 
		  (T(\id, \beta, \gamma))_*;(T(\pi_2^{M'}, \id, \id))^*;(T(\id, \pi_2^{N'}, \pi_2^{R'}))_* \\
		\leq & (T(\id, \pi_1^{N'}, \pi_1^{R'}))^*;
		(T(\pi_1^{M'}, \id, \id))_* ;(T(\pi_2^{M'}, \id, \id))^*;(T(\id, \pi_2^{N'}, \pi_2^{R'}))_* \\
		=&	\mathcal{L}_1(M',N',R')\\
	\end{aligned}\]

	\item for $M: A_1 \profto A_2$, $M' :A_2 \profto A_3$, we have:
	\begin{center}
		\begin{tikzpicture}[thick,xscale=0.7, yscale=0.7]
			\node (A) at (2.5,0) {$A_1$};
			\node (Y1) at (4,1.5) {$\graph{M}$};
			\node (B) at (5.5,0) {$A_2$};
			\draw [->] (Y1)  to node [above left] {$\pi_1^{M}$} (A);
			\draw [->] (Y1)  to node [below left] {$\pi_2^M$} (B);
			
			\node (X2) at (7,1.5) {$\graph{M'}$};
			\node (P) at (5.5,3) {$P_{M,M'}$};
			\node (C) at (8.5,0) {$A_3$};
			\draw [->] (P) to node [above left] {$p_M$} (Y1);
			\draw [->] (P) to node [above right] {$p_{M'}$} (X2);
			\draw [->] (X2)  to node [below right] {$\pi_1^M$} (B);
			\draw [->] (X2)  to node [above right] {$\pi_2^{M'}$} (C);
			
			\node at (5.5,1.5) {pb};
			
			\begin{scope}[xshift=10cm]
				\node (A) at (2.5,0) {$A_1$};
				\node (Y1) at (3.5,1.5) {$\graph{M}$};
				\node (B) at (5.5,0.8) {$\graph{M; M'}$};
				\draw [->] (Y1)  to node [above left] {$\pi_1^M$} (A);
				\draw [->] (B)  to node [below] {$\pi_1^{M; M'}$} (A);
				
				\node (X2) at (7.5,1.5) {$\graph{M'}$};
				\node (P) at (5.5,3) {$P_{M; M'}$};
				\node (C) at (8.5,0) {$A_3$};
				\draw [->] (P) to node [above left] {$p_{M'}$} (Y1);
				\draw [->, dotted] (P) to node [] {$u_{M,M'}$} (B);
				\draw [->] (P) to node [above right] {$p_{M'}$} (X2);
				\draw [->] (B)  to node [below] {$\pi_1^{M; M'}$} (C);
				\draw [->] (X2)  to node [above right] {$\pi_2^{M'}$} (C);
				
			\end{scope}
		\end{tikzpicture}
	\end{center}
		\[\begin{aligned}
	&\mathcal{L}_1(M,\id_C, \id_X);	\mathcal{L}_1(M',\id_C, \id_X)\\
	&= 	
	(T(\pi_1^{M'}, \id, \id))_*  ;(T(\pi_2^{M'}, \id, \id))^*;(T(\pi_1^{M'}, \id, \id))_*  ;(T(\pi_2^{M'}, \id, \id))^* \\
	&\leq (T(\pi_1^{M'}, \id, \id))_*  ;(T(p_M, \id, \id))_*;(T(p_{M'}, \id, \id))^*  ;(T(\pi_2^{M'}, \id, \id))^* \\
	&=(T(\pi_1^{M; M'}, \id, \id))_*  ;(T(u_{M,M'}, \id, \id))_*;(T(u_{M,M'}, \id, \id))^*  ;(T(\pi_2^{M; M'}, \id, \id))^* \\
	&=(T(\pi_1^{M; M'}, \id, \id))_*   ;(T(\pi_2^{M; M'}, \id, \id))^* =	\mathcal{L}_1(M; M',\id_C, \id_X)
\end{aligned}\]
	\item the proof for $\mathcal{L}_1(\id ,N'\circ N,R'\circ R) \leq \mathcal{L}_1(\id ,N',R') \circ \mathcal{L}_1(\id ,N,R)  $	is similar.

		\item for $f : A\to B$, $g : C \to D$ and $h : X \to Y$, we have  
		 \begin{center}
			\begin{tikzpicture}[thick,scale=0.5]
				\node (A) at (0,1.5) {\small$B$};
				\node (B) at (2,3) {\small$\graph{f^*}$};
				\node (C) at (4,1.5) {\small$A$};
				\node (D) at (2,0) {\small$A$};
				
				\draw [->] (B)  to  node [above left] {\scriptsize$\pi_1^{f^*}$} (A);
				\draw [->] (B)  to  node [above right] {\scriptsize$\pi_2^{f^*}$} (C);	
				\draw [->] (D) to  node [below left] {\scriptsize$f$} (A);
				\draw [double=] (D)  to node [below right] {} (C);
				\draw [->] (B) to  node [left] {\scriptsize$\pi_2^{f^*}$} (D);
				
				\begin{scope}[xshift=7cm]
			\node (A) at (0,1.5) {\small$C$};
\node (B) at (2,3) {\small$\graph{g_*}$};
\node (C) at (4,1.5) {\small$D$};
\node (D) at (2,0) {\small$C$};

\draw [->] (B)  to  node [above left] {\scriptsize$\pi_1^{g_*}$} (A);
\draw [->] (B)  to  node [above right] {\scriptsize$\pi_2^{g_*}$} (C);	
\draw [double=] (D) to  node [below left] {} (A);
\draw [->] (D)  to node [below right] {\scriptsize$g$} (C);
\draw [->] (B) to  node [left] {\scriptsize$\pi_1^{g_*}$} (D);
				\end{scope}
							\begin{scope}[xshift=14cm]
								\node (A) at (0,1.5) {\small$X$};
					\node (B) at (2,3) {\small$\graph{h_*}$};
					\node (C) at (4,1.5) {\small$Y$};
					\node (D) at (2,0) {\small$X$};
					
					\draw [->] (B)  to  node [above left] {\scriptsize$\pi_1^{h_*}$} (A);
					\draw [->] (B)  to  node [above right] {\scriptsize$\pi_2^{h_*}$} (C);	
					\draw [double=] (D) to  node [below left] {} (A);
					\draw [->] (D)  to node [below right] {\scriptsize$h$} (C);
					\draw [->] (B) to  node [left] {\scriptsize$\pi_1^{h_*}$} (D);
				\end{scope}
			\end{tikzpicture}
		\end{center}
		where $\pi_2^{f^*}$, $\pi_1^{g_*}$ and $\pi_1^{h_*}$ are split epi. We obtain:
			\[\begin{aligned}
			&\mathcal{L}_1(f^*, g_*, h_* )= (T(\id, \pi_1^{g_*}, \pi_1^{h_*}))^* ; (T(\pi_1^{f^*}, \id, \id))_* ;(T(\pi_2^{f^*}, \id, \id))^*	;(T(\id, \pi_2^{g_*}, \pi_2^{h_*}))_*  \\
			&= 	(T(\id, \pi_1^{g_*}, \pi_1^{h_*}))^*;(T(f, \id, \id))_* ;(T(\pi_2^{f^*}, \id, \id))_* ;\\
			&\qquad\qquad\qquad(T(\pi_2^{f^*}, \id, \id))^* ;(T(\id, \pi_1^{g_*}, \pi_1^{h_*}))_* ;T(\id, g, h))_* \\
			&= (T(\id, \pi_1^{g_*}, \pi_1^{h_*}))^*;(T(f, \id, \id))_* ;(T(\id, \pi_1^{g_*}, \pi_1^{h_*}))_* ;(T(\id, g, h))_* \\
			&= (T(\id, \pi_1^{g_*}, \pi_1^{h_*}))^*;(T(\id, \pi_1^{g_*}, \pi_1^{h_*}))_*;(T(f, \id, \id))_*  ;(T(\id, g, h))_* \\
			&= (T(f, g, h))_*
		\end{aligned}\]
		The proof for $	(T(f, g, h))^*=  \mathcal{L}_1(f_*,  g^*, h^* )$ is similar.
		\item for $M : A \profto B$, $N: C \profto D$ and $R : X \profto Y$, we use that the graph of the transpose switches source and target in spans:
		\begin{center}
			\begin{tikzpicture}[thick,scale=0.5]
			\node (A) at (0,1.5) {\small$B$};
			\node (B) at (2,3) {\small$\graph{M^t}$};
			\node (C) at (4,1.5) {\small$A$};
			\node (D) at (2,0) {\small$\graph{M}$};
			
			\draw [->] (B)  to  node [above left] {\scriptsize$\pi_1^{M^t}$} (A);
			\draw [->] (B)  to  node [above right] {\scriptsize$\pi_2^{M^t}$} (C);	
			\draw [->] (D) to  node [below left] {\scriptsize$\pi_2^{M}$} (A);
			\draw [->] (D)  to node [below right] {\scriptsize$\pi_1^{M}$} (C);
			\draw [->] (B) to  node [left] {\scriptsize$\cong$} (D);
			
			\begin{scope}[xshift=6cm]
							\node (A) at (0,1.5) {\small$D$};
				\node (B) at (2,3) {\small$\graph{N^t}$};
				\node (C) at (4,1.5) {\small$C$};
				\node (D) at (2,0) {\small$\graph{N}$};
				
				\draw [->] (B)  to  node [above left] {\scriptsize$\pi_1^{N^t}$} (A);
				\draw [->] (B)  to  node [above right] {\scriptsize$\pi_2^{N^t}$} (C);	
				\draw [->] (D) to  node [below left] {\scriptsize$\pi_2^{N}$} (A);
				\draw [->] (D)  to node [below right] {\scriptsize$\pi_1^{N}$} (C);
				\draw [->] (B) to  node [left] {\scriptsize$\cong$} (D);
			\end{scope}
			
			\begin{scope}[xshift=12cm]
							\node (A) at (0,1.5) {\small$Y$};
				\node (B) at (2,3) {\small$\graph{R^t}$};
				\node (C) at (4,1.5) {\small$X$};
				\node (D) at (2,0) {\small$\graph{R}$};
				
				\draw [->] (B)  to  node [above left] {\scriptsize$\pi_1^{R^t}$} (A);
				\draw [->] (B)  to  node [above right] {\scriptsize$\pi_2^{R^t}$} (C);	
				\draw [->] (D) to  node [below left] {\scriptsize$\pi_2^{R}$} (A);
				\draw [->] (D)  to node [below right] {\scriptsize$\pi_1^{R}$} (C);
				\draw [->] (B) to  node [left] {\scriptsize$\cong$} (D);
			\end{scope}
		\end{tikzpicture}
		\end{center}
		We obtain:
		\[	\begin{aligned}
			\mathcal{L}_1(M^t,N^t,R^t) &=(T(\id, \pi_2^N, \pi_2^R))^* ;
			(T(\pi_2^M, \id, \id))_* ;(T(\pi_1^M, \id, \id))^* ;(T(\id, \pi_1^N, \pi_1^R))_* \\
			&=\left((T(\id, \pi_1^N, \pi_1^R))^* ;(T(\pi_1^M, \id, \id))_* ;(T(\pi_2^M, \id, \id))^*;(T(\id, \pi_2^N, \pi_2^R))_*  \right)^t\\
			&= 	(\mathcal{L}_1(M,N,R))^t
		\end{aligned}\]

		\item let $\mathcal{L}$ be an abitrary lax extension for $T$, for relations $M : A \profto B$, $N: C \profto D$ and $R : X \profto Y$, we have:
\begin{align*}
			&\mathcal{L}_1(M,N,R)=(T(\id, \pi_1^N, \pi_1^R))^*;
			(T(\pi_1^M, \id, \id))_*  ;(T(\pi_2^M, \id, \id))^* ;(T(\id, \pi_2^N, \pi_2^R))_* \\
			&\leq \mathcal{L}(\id, (\pi_1^N)^*, (\pi_1^R)^*) ; \mathcal{L}((\pi_1^M)^* \id, \id) ; \mathcal{L}((\pi_2^M)_*, \id, \id);\mathcal{L}(\id, (\pi_2^N)_*, (\pi_2^R)_*) \\
			&\leq \mathcal{L}((\pi_1^M)^*;(\pi_2^M)_*, (\pi_1^N)^*; (\pi_2^N)_* (\pi_1^R)^*;(\pi_2^R)_*) = \mathcal{L}(M,N,R)\qedhere
		\end{align*}
\end{enumerate}
\end{proof}
\begin{proof}[Proof of Proposition~\ref{prop:laxExtParam}] It is straightforward to obtain $\mathcal{L}_1=\mathcal{L}_2=\mathcal{L}_3=\mathcal{L}_4$ from the strong interchange property. We show that $\mathcal{L}_1$ satisfies axiom $(2)$ in Definition~\ref{def:ParamLaxExtensionMonadSet}:
		\item For relations $M: A_1 \profto A_2$, $M' :A_2 \profto A_3$,  $N: B_1 \profto B_2$, $N' :B_2 \profto B_3$, $R: X_1 \profto X_2$ and $R' :X_2 \profto X_3$, we have:
			\[	\begin{aligned}
			&\mathcal{L}_1(M,N,R);  \mathcal{L}_1(M,'N',R')\\
			=&(T(\id, \pi_1^N, \pi_1^R))^*;
			(T(\pi_1^M, \id, \id))_*  ;(T(\pi_2^M, \id, \id))^* ;(T(\id, \pi_2^N, \pi_2^R))_* \\
			&\qquad;(T(\id, \pi_1^{N'}, \pi_1^{R'}))^*;
			(T(\pi_1^{M'}, \id, \id))_*  ;(T(\pi_2^{M'}, \id, \id))^* ;(T(\id, \pi_2^{N'}, \pi_2^{R'}))_* \\
			=&(T(\id, \pi_1^N, \pi_1^R))^*;
			(T(\pi_1^M, \id, \id))_*  ;(T(\pi_2^M, \id, \id))^* 
			;(T(\id, p_N, p_R))^*\\
			&\qquad ;(T(\id, p_{N'}, p_{R'}))_*;
			(T(\pi_1^{M'}, \id, \id))_*  ;(T(\pi_2^{M'}, \id, \id))^* ;(T(\id, \pi_2^{N'}, \pi_2^{R'}))_* \\
			=&T(\id, \pi_1^N, \pi_1^R))^*;
			(T(\pi_1^M, \id, \id))_*  
			;(T(\id, p_N, p_R))^* ;
						(T(\pi_2^M, \id, \id))^*\\
						&\qquad ;
			(T(\pi_1^{M'}, \id, \id))_*  ;
				(T(\id, p_{N'}, p_{R'}))_*;
			(T(\pi_2^{M'}, \id, \id))^* ;(T(\id, \pi_2^{N'}, \pi_2^{R'}))_* \\
			\leq &T(\id, \pi_1^N, \pi_1^R))^*;
			(T(\pi_1^M, \id, \id))_*  
			;(T(\id, p_N, p_R))^* ;
				(T(p_M, \id, \id))_*\\
				&\qquad ; 	(T(p_{M'}, \id, \id))^* ;
			(T(\id, p_{N'}, p_{R'}))_*;
			(T(\pi_2^{M'}, \id, \id))^* ;(T(\id, \pi_2^{N'}, \pi_2^{R'}))_* \\
				= &T(\id, \pi_1^N, \pi_1^R))^*;
			(T(\id, p_N, p_R))^* ;
			(T(\pi_1^M, \id, \id))_* ;
			(T(p_M, \id, \id))_* \\
			&\qquad; 	(T(p_{M'}, \id, \id))^* ;
			(T(\pi_2^{M'}, \id, \id))^* ;
			(T(\id, p_{N'}, p_{R'}))_*;
			(T(\id, \pi_2^{N'}, \pi_2^{R'}))_* \\
			=& 	(T(\id, \pi_1^{N;N'}, \pi_1^{R;R'}))^*;(T(\pi_1^{M; M'}, \id, \id))_*   ;(T(\pi_2^{M; M'}, \id, \id))^* ; 	(T(\id, \pi_2^{N;N'}, \pi_2^{R;R'}))_*
		\end{aligned}\]
					We use that $T$ preserves pullbacks in the last two variables for the second equality and that the interchange squares are weak pullbacks for the penultimate equality. \qedhere
\end{proof}

\begin{proof}[Proof of Proposition~\ref{prop:laxExtParamWeak}]
	We show that $\mathcal{L}_4$ satisfies axiom $(2)$ in Definition~\ref{def:ParamLaxExtensionMonadSet}:
		\item For relations $M: A_1 \profto A_2$, $M' :A_2 \profto A_3$,  $N: B_1 \profto B_2$, $N' :B_2 \profto B_3$, $R: X_1 \profto X_2$ and $R' :X_2 \profto X_3$, we have:
		\[	\begin{aligned}
			&\mathcal{L}_4(M,N,R); \mathcal{L}_4(M,'N',R')\\
			=&	(T(\pi_1^M, \id, \id))_*  ;(T(\id, \pi_1^N, \pi_1^R))^*;(T(\id, \pi_2^N, \pi_2^R))_* ;
			(T(\pi_2^M, \id, \id))^* \\
			&\qquad;(T(\pi_1^{M'}, \id, \id))_*  ;(T(\id, \pi_1^{N'}, \pi_1^{R'}))^*;
		(T(\id, \pi_2^{N'}, \pi_2^{R'}))_*;	(T(\pi_2^{M'}, \id, \id))^*  \\
		\leq &	(T(\pi_1^M, \id, \id))_*  ;(T(\id, \pi_1^N, \pi_1^R))^*;(T(\id, \pi_2^N, \pi_2^R))_* ;
		(T(p_M, \id, \id))_* \\
		&\qquad;(T(p_{M'}, \id, \id))^*  ;(T(\id, \pi_1^{N'}, \pi_1^{R'}))^*;
		(T(\id, \pi_2^{N'}, \pi_2^{R'}))_*;	(T(\pi_2^{M'}, \id, \id))^*  \\
		= &	(T(\pi_1^M, \id, \id))_*  ;(T(\id, \pi_1^N, \pi_1^R))^*;(T(p_M, \id, \id))_* ;(T(\id, \pi_2^N, \pi_2^R))_* ;
		\\
		&\qquad ;(T(\id, \pi_1^{N'}, \pi_1^{R'}))^*;(T(p_{M'}, \id, \id))^* ;
		(T(\id, \pi_2^{N'}, \pi_2^{R'}))_*;	(T(\pi_2^{M'}, \id, \id))^*  \\
		\leq &	(T(\pi_1^M, \id, \id))_* ;(T(p_M, \id, \id))_*  ;(T(\id, \pi_1^N, \pi_1^R))^*;(T(\id, \pi_2^N, \pi_2^R))_* ;
		\\
		&\qquad ;(T(\id, \pi_1^{N'}, \pi_1^{R'}))^* ;
		(T(\id, \pi_2^{N'}, \pi_2^{R'}))_*;	(T(p_{M'}, \id, \id))^*;(T(\pi_2^{M'}, \id, \id))^*  \\
		= &	(T(\pi_1^{M; M'}, \id, \id))_*  ;(T(u_{M,M'}, \id, \id))_*;
		T(\id, \pi_1^{N;N'}, \pi_1^{R;R'}))_*	\\
		&\qquad ;	T(\id, \pi_2^{N;N'}, \pi_2^{R;R'}))_*
		(T(u_{M,M'}, \id, \id))^*  ;
		(T(\pi_2^{M; M'}, \id, \id))^*\\
		= &	(T(\pi_1^{M; M'}, \id, \id))_*  ;
		T(\id, \pi_1^{N;N'}, \pi_1^{R;R'}))_*;(T(u_{M,M'}, \id, \id))_*	\\
		&\qquad ;	
		(T(u_{M,M'}, \id, \id))^*  ; T(\id, \pi_2^{N;N'}, \pi_2^{R;R'}))_*;
		(T(\pi_2^{M; M'}, \id, \id))^*\\
		= 	&	(T(\pi_1^{M; M'}, \id, \id))_*  ;
		T(\id, \pi_1^{N;N'}, \pi_1^{R;R'}))_*; T(\id, \pi_2^{N;N'}, \pi_2^{R;R'}))_*;
		(T(\pi_2^{M; M'}, \id, \id))^*\\
		=&	\mathcal{L}_4(M; M',N;N', R;R')
		\end{aligned}\]
		We use that $T$ preserves pullbacks in the last two variables for the third equality and the weak interchange property for the fourth equality. \qedhere
\end{proof}

\begin{proof}[Proof of Proposition~\ref{prop:LaxParamExtensionMonad}]
	\begin{itemize}
		\item[]
		\item Unit: for all $M :A \profto B$ and $R : X \profto Y$. we have
		\[\begin{aligned}
		R ;	(\eta_{B, Y})_* &=(\pi_1^R)^* ; (\pi_2^R)_* ; (\eta_{B, Y})_* \\
		&\leq (\pi_1^R)^* ; (\eta_{\graph{M}, \graph{R}})_*; (\eta_{\graph{M}, \graph{R}})^*;(\pi_2^R)_* ; (\eta_{B, Y})_*\\
		&\leq  (\eta_{A, X})_*;(T(\pi_1^M, \id, \id))_* ;(T(\id, \pi_1^M, \pi_1^R))^* ;(T(\id, \pi_2^M, \pi_2^R))_* ;(T(\pi_2^M, \id, \id))^*\\
		&= (\eta_{A, X})_*;	\mathcal{L}_4(M,M,R)
		\end{aligned}\]
		\item Multiplication: we want to show that for all $M: A_1 \profto B_1$, $N: A_2 \profto B_2$, $O: A_3 \profto B_3$ and $R : X \profto Y$, the following holds:
		\[	\mathcal{L}_4(M,N,\mathcal{L}_4(N, O, R)) ;(\mu_{B_1, B_2, B_3,Y})_*
	 \leq (\mu_{A_1, A_2,A_3,X})_*;\mathcal{L}_4(M,O, R)
		\]
	Since $\mathcal{L}_4$ is a lax extension, the two variable operator $\mathcal{L}_4(\id, -, -)$ preserves composition strictly. Indeed, we have $\mathcal{L}_4(\id, N; N', R; R') \geq \mathcal{L}_4(\id, N, R) ;\mathcal{L}_4(\id, N', R')$ because it is a lax extension, and by construction, we always have $\mathcal{L}_4(\id, N; N', R; R') \leq \mathcal{L}_4(\id, N, R) ;\mathcal{L}_4(\id, N', R')$ (Theorem~\ref{th:paramBarr}.4). 
	Therefore, using Lemma~\ref{lem:laxdblefunctor}, the desired inequality above becomes:
	\[\begin{aligned}
	 &(T(\pi_1, \id, T(\pi_1, \id, \id)))_* ;(T(\id, \pi_1, T(\id, \pi_1,\pi_1)))^* ;(T(\id, \pi_2,T( \id, \pi_2, \pi_2)))_* ;\\
	 &\qquad \qquad \qquad \qquad(T(\pi_2, \id,T( \pi_2, \id, \id)))^*;(\mu_{B_1, B_2, B_3,Y})_* \\
	 &\leq (\mu_{A_1, A_2,A_3,X})_*;(T(\pi_1, \id, \id))_* ;(T(\id, \pi_1, \pi_1))^* ;(T(\id, \pi_2, \pi_2))_* ;(T(\pi_2, \id, \id))^*
	\end{aligned}\]
	We will combine the following two inequalities to establish the desired inequality above:
	\begin{align*}
&(T(\pi_1, \id, T(\pi_1, \id, \id)))_* ;(T(\id, \pi_1, T(\id, \pi_1,\pi_1)))^* ;(\mu_{\graph{M}, \graph{N},\graph{O},\graph{R}})_* \\
&\leq (\mu_{A_1, A_2,A_3,X})_*;(T(\pi_1, \id, \id))_* ;(T(\id, \pi_1, \pi_1))^* \tag{$\star$} \label{ineq1} 
	\end{align*}
	\begin{align*}
		&(\mu_{\graph{M}, \graph{N},\graph{O},\graph{R}})^*; (T(\id, \pi_2,T( \id, \pi_2, \pi_2)))_* ;(T(\pi_2, \id,T( \pi_2, \id, \id)))^*;\\
		& \qquad(\mu_{B_1, B_2, B_3,Y})_* \\
	&\leq	(T(\id, \pi_2, \pi_2))_* ;(T(\pi_2, \id, \id))^* \tag{$\star\star$} \label{ineq2} 
	\end{align*}
	\begin{itemize}
		\item proof of (\ref{ineq1}) 
		\[\begin{aligned}
			&(T(\pi_1, \id, T(\pi_1, \id, \id)))_* ;(T(\id, \pi_1, T(\id, \pi_1,\pi_1)))^* ;(\mu_{\graph{M}, \graph{N},\graph{O},\graph{R}})_* \\
			=&(T(\pi_1, \id, \id))_* ; (T(\id, \id, T(\pi_1, \id, \id)))_*		;(T(\id, \pi_1, \id))^*	;(T(\id, \id, T(\id, \pi_1,\pi_1)))^* ;\\
			&\qquad (\mu_{\graph{M}, \graph{N},\graph{O},\graph{R}})_* \\
			=&(T(\pi_1, \id, \id))_* 	;(T(\id, \pi_1, \id))^*	; (T(\id, \id, T(\pi_1, \id, \id)))_*	;(T(\id, \id, T(\id, \pi_1,\pi_1)))^* ;\\
			&\qquad (\mu_{\graph{M}, \graph{N},\graph{O},\graph{R}})_* \\
			=\leq&(T(\pi_1, \id, \id))_* 	;(T(\id, \pi_1, \id))^*	; (T(\id, \id, T(\pi_1, \id, \id)))_*	; (\mu_{\graph{M}, \graph{N},A_3,X})_*  ; \\
			&\qquad (T(\id, \pi_1,\pi_1))^* \\
			\leq&(T(\pi_1, \id, \id))_* 	;(\mu_{\graph{M},A_2,A_3,X})_*; (T(\id, \pi_1,\pi_1))^* \\
			\leq & (\mu_{A_1, A_2,A_3,X})_*;(T(\pi_1, \id, \id))_* ;(T(\id, \pi_1, \pi_1))^*
		\end{aligned}
		\]
		\item proof of (\ref{ineq2}) 
		\[\begin{aligned}
			&(\mu_{\graph{M}, \graph{N},\graph{O},\graph{R}})^*; (T(\id, \pi_2,T( \id, \pi_2, \pi_2)))_* ;(T(\pi_2, \id,T( \pi_2, \id, \id)))^*;\\
			& \qquad(\mu_{B_1, B_2, B_3,Y})_* \\
			=& (\mu_{\graph{M}, \graph{N},\graph{O},\graph{R}})^*; 	(T(\id, \id,T( \id, \pi_2, \pi_2)))_* ;	(T(\id, \pi_2,\id))_* ;\\
			& \qquad 	(T(\id ,\id,T( \pi_2, \id, \id)))^*;	(T(\pi_2, \id,\id)))^*;(\mu_{B_1, B_2, B_3,Y})_* \\
			=& (\mu_{\graph{M}, \graph{N},\graph{O},\graph{R}})^*; 	(T(\id, \id,T( \id, \pi_2, \pi_2)))_* ;	(T(\id, \id,T( \pi_2, \id, \id)))^*;\\
			& \qquad 	(T(\id, \pi_2,\id))_*	;	(T(\pi_2, \id,\id)))^*;(\mu_{B_1, B_2, B_3,Y})_* \\
			\leq& (\mu_{\graph{M}, \graph{N},\graph{O},\graph{R}})^*; 	(T(\id, \id,T( \id, \pi_2, \pi_2)))_* ;	(T(\id, \id,T( \pi_2, \id, \id)))^*;\\
			& \qquad 	(T(\id, \pi_2,\id))_*	;(\mu_{\graph{M}, B_2, B_3,Y})_* ;	(T(\pi_2, \id,\id))^*\\
			\leq& (\mu_{\graph{M}, \graph{N},\graph{O},\graph{R}})^*; 	(T(\id, \id,T( \id, \pi_2, \pi_2)))_* ;	(\mu_{\graph{M}, \graph{N}, B_3,Y})_*;	(T(\pi_2, \id,\id))^*\\
			\leq&	(T(\id, \pi_2, \pi_2))_* ;(T(\pi_2, \id, \id))^* \
		\end{aligned}
		\]
	\end{itemize}
	In the proofs of (\ref{ineq1}) and (\ref{ineq2}), the second equality uses the fact that $\mathcal{L}_4(\id, -, -)$ preserves composition strictly. 
	\item Strength: for all $M: A_1 \profto B_1$, $N: A_2 \profto B_2$, $R : X_1 \profto Y_1$ and $S : X_2 \profto Y_2$, we have 
	\[\begin{aligned}
			&(R\times \mathcal{L}_4(M,N, S)) ; (\strength_{A_2, B_2, X_2, Y_2})_*\\
			=& (\id \times (T(\pi_1, \id, \id))_* ) ;((\pi_1)^* \times (T(\id, \pi_1, \pi_1))^*);((\pi_2)_* \times T(\id, \pi_2, \pi_2))_* ;\\
			&\qquad (\id \times (T(\pi_2, \id, \id))^*);( \strength_{A_2, B_2, X_2, Y_2})_* \\
			\leq& (\id \times (T(\pi_1, \id, \id))_* ) ;((\pi_1)^* \times (T(\id, \pi_1, \pi_1))^*);((\pi_2)_* \times T(\id, \pi_2, \pi_2))_* ;\\
			&\qquad ( \strength_{\graph{M}, B_2, X_2, Y_2})_* ;(T(\pi_2, \id, \id\times \id))^*\\
			\leq& (\id \times (T(\pi_1, \id, \id))_* ) ;((\pi_1)^* \times (T(\id, \pi_1, \pi_1))^*);( \strength_{\graph{M}, \graph{N}, \graph{R}, \graph{S}})_*; \\
			&\qquad ( T(\id, \pi_2, \pi_2 \times \pi_2))_*;(T(\pi_2, \id, \id\times \id))^*\\
			\leq& (\id \times (T(\pi_1, \id, \id))_* ) ; ( \strength_{\graph{M},B_1,X_1,Y_1})_* ;	(T(\id, \pi_1, \pi_1 \times \pi_1))^*; \\
			&\qquad ( T(\id, \pi_2, \pi_2 \times \pi_2))_*;(T(\pi_2, \id, \id\times \id))^*\\
			\leq& ( \strength_{A_1,B_1,X_1,Y_1})_*; (T(\pi_1 , \id, \id \times \id))_*   ;	(T(\id, \pi_1, \pi_1 \times \pi_1))^*; \\
&\qquad ( T(\id, \pi_2, \pi_2 \times \pi_2))_*;(T(\pi_2, \id, \id\times \id))^*\\
=& (\strength_{A_1, B_1, X_1, Y_1})_*;  \mathcal{L}_4(M,N, R \times S).
	\end{aligned}\]
		\end{itemize}
\end{proof}
\section{Proofs of Section~\ref{sec:bisim}}

\begin{proof}[Proof of Lemma~\ref{lem:AMLaxExtSoundness}]

	Assume that we have an AM-bisimulation w.r.t $(M,N)$ as in Section~\ref{sec:AMbisim}, then we obtain that it is a $\Barr{T}$-bisimulation w.r.t $(M,N)$ from the axioms of parametrized lax extension:
	\begin{align*}
		&(\pi_1^R)^*; 	(\pi_2^R)_* \\
		&\leq\alpha_* ;
		(T(\pi_1^{M}, \id, \id))_*  ;(T(\id, \pi_1^N, \pi_1^R))^*;(T(\id, \pi_2^N, \pi_2^R))_*;(T(\pi_2^{M}, \id, \id))^* ; \beta^* \\
		&\leq \alpha_* ;
		\Barr{T}((\pi_1^{M})^*, \id, \id)  ;	\Barr{T}(\id, (\pi_1^N)^*, (\pi_1^R)^*);	\Barr{T}(\id, (\pi_2^N)_*, (\pi_2^R)_*);\Barr{T}((\pi_2^{M})_*, \id, \id) ; \alpha^*\\
		&\leq \alpha_* ; \Barr{T}(M,N,R); \beta^* \qedhere
	\end{align*}
\end{proof}

\begin{proof}[Proof of Lemma~\ref{lem:paramBehEquivFinal}]
	\begin{itemize}
		\item[]
		\item Assume that $(x,y)$ are behaviorally equivalent. Let $\sem{-}_3: Z \to \Omega$ be the unique coalgebra morphism from $\gamma$ to $\omega$. By uniqueness, we obtain $\sem{-}_1 = \sem{-}_3 \circ f$ and $\sem{-}_2= \sem{-}_3\circ g$. Since $(x,y) \in \graph{R}$ and $f(x) = g(y)$, we obtain $\sem{x}_1 = \sem{y}_2$.
		\item For the other direction, assume that $\sem{x}_1 = \sem{y}_2$. Let $(P, p_1, p_2)$ be the following pullback:
			\begin{center}
			\begin{tikzpicture}[thick,scale=0.5]
			
					\node (A) at (0,1.5) {$X$};
					\node (B) at (2,0) {$Z$};
					\node (C) at (4,1.5) {$X$};
					\node (D) at (2,3) {$P$};
					
					\draw [->] (A)  to  node [below left] {\scriptsize$f$} (B);
					\draw [->] (C)  to  node [below right] {\scriptsize$g$} (B);
					\draw [->] (D) to  node [above left] {\scriptsize$p_1$} (A);
					\draw [->] (D)  to node [above right] {\scriptsize$p_2$} (C);
					\node at (2,1.5) {\scriptsize$\mathrm{pb}$};

			\end{tikzpicture}
		\end{center}
		and define $R := f_*; g^*$. From the universal property of tabulators, we obtain that the span $(P, p_1, p_2)$ is isomorphic to $(\graph{R}, \pi_1^R, \pi_2^R)$ which finishes the proof.\qedhere
	\end{itemize}
\end{proof}

\begin{proof}[Proof of Proposition~\ref{prop:AMBehEquivsoundness}]
	Assume that
	$N: C \profto C$ is difunctional \ie that there exists $h : C \to W$ and $k : B \to W$ such that the left square below is a weak-pullback
	\begin{center}
		\begin{tikzpicture}[thick,scale=0.5]
			\node (A) at (0,1.5) {$C$};
			\node (B) at (2,0) {$W$};
			\node (C) at (4,1.5) {$C$};
			\node (D) at (2,3) {$\graph{N}$};
			
			\draw [->] (A)  to  node [below left] {\scriptsize$h$} (B);
			\draw [->] (C)  to  node [below right] {\scriptsize$k$} (B);
			\draw [->] (D) to  node [above left] {\scriptsize$\pi_1^N$} (A);
			\draw [->] (D)  to node [above right] {\scriptsize$\pi_2^N$} (C);
			\node at (2,1.5) {\scriptsize$\mathrm{wpb}$};
			
			\begin{scope}[xshift=7cm]
				\node (A) at (0,1.5) {$X$};
				\node (B) at (2,0) {$Z$};
				\node (C) at (4,1.5) {$X$};
				\node (D) at (2,3) {$\graph{R}$};
				
				\draw [->] (A)  to  node [below left] {\scriptsize$f$} (B);
				\draw [->] (C)  to  node [below right] {\scriptsize$g$} (B);
				\draw [->] (D) to  node [above left] {\scriptsize$\pi_1^R$} (A);
				\draw [->] (D)  to node [above right] {\scriptsize$\pi_2^R$} (C);
				\node at (2,1.5) {\scriptsize$\mathrm{wpb}$};
			\end{scope}
		\end{tikzpicture}
	\end{center}
	By definition, a relation $R : X\profto X$ is a \emph{behavioral equivalence w.r.t $(M,N)$} if there exists $f : X \to Z$ and  $g : X \to Z$ and a coalgebra $\zeta : Z \to  T(\graph{M},W,Z)$ such that the right square above is a weak-pullback and the following two squares below commute:

	\begin{center}
		\begin{tikzpicture}[thick, xscale=0.8, yscale=0.7]
			\node (A) at (0,0) {\small$T(A,C,X)$};
			\node (B) at (3,-1) {\small$T(\graph{M},W,Z)$};
			\node (C) at (0,2) {\small$X$};
			\node (D) at (3,1) {\small$Z$};
			\node (E) at (6,2) {\small$X$};
			\node (F) at (6,0) {\small$T(A,C,X)$};
			
			\draw [->] (C) -- node [left] {\scriptsize$\alpha$} (A);
			\draw [->] (A) -- node [below left] {\scriptsize$T(\pi_1^M, h,f)$} (B);
			\draw [->] (F) -- node [below right] {\scriptsize$T(\pi_2^M, k, g)$} (B);
			\draw [->] (C) -- node [above] {\scriptsize$f$} (D);
			\draw [->] (E) -- node [above] {\scriptsize$g$} (D);
			\draw [->] (E) -- node [right] {\scriptsize$\alpha$} (F);
			\draw [->] (D) -- node [right] {\scriptsize$\zeta$} (B);
		\end{tikzpicture}
	\end{center}
	We want to show that $R$ is an AM-bisimulation w.r.t $(M,N)$. Since $T$ preserves weak pullbacks in the last two variables, the square below is a weak pullback.
	\begin{center}
		\begin{tikzpicture}[thick,xscale=0.7, yscale=0.8]
			\node (A) at (0,1.5) {\scriptsize$T(\graph{M}, C,X)$};
			\node (B) at (2,0) {\scriptsize$T(\graph{M}, W,Z)$};
			\node (C) at (4,1.5) {\scriptsize$T(\graph{M}, C,X)$};
			\node (D) at (2,3) {\scriptsize$T(\graph{M}, \graph{N},\graph{R})$};
			
			\draw [->] (A)  to  node [below left] {\scriptsize$T(\id, h,f)$} (B);
			\draw [->] (C)  to  node [below right] {\scriptsize$T(\id, k,g)$} (B);
			\draw [->] (D) to  node [above left] {\scriptsize$T(\id, \pi_1^N,  \pi_1^R)$} (A);
			\draw [->] (D)  to node [above right] {\scriptsize$T(\id, \pi_2^N,  \pi_2^R)$} (C);
			\node at (2,1.5) {\scriptsize$\mathrm{wpb}$};
			
		\end{tikzpicture}
	\end{center}
	Note that we see here why both $N$ and $R$ are assumed to be difunctional.
	By the universal property of weak-pullbacks, since the diagram below commutes,
		
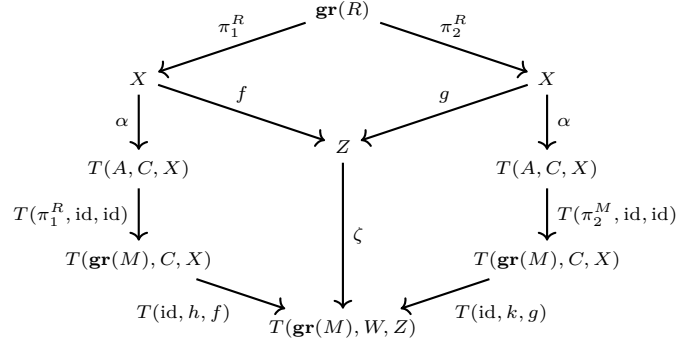
\begin{center}
		\begin{tikzpicture}[thick, xscale=0.9, yscale=0.6]
			\node (A) at (0,0) {\scriptsize$T(A,C,X)$};
			\node (A1) at (0,-2) {\scriptsize$T(\graph{M},C,X)$};
			\node (B) at (3,-3.5) {\scriptsize$T(\graph{M},W,Z)$};
			\node (C) at (0,2) {\scriptsize$X$};
			\node (D0) at (3,3.5) {\scriptsize$\graph{R}$};
			\node (D) at (3,0.5) {\scriptsize$Z$};
			\node (E) at (6,2) {\scriptsize$X$};
			\node (F) at (6,0) {\scriptsize$T(A,C,X)$};
			\node (F1) at (6,-2) {\scriptsize$T(\graph{M},C,X)$};
			
			\draw [->] (C) -- node [left] {\scriptsize$\alpha$} (A);
			\draw [->] (A1) -- node [below left] {\scriptsize$T(\id, h,f)$} (B);
			\draw [->] (A) -- node [left] {\scriptsize$T(\pi_1^R, \id, \id)$} (A1);
			\draw [->] (F) -- node [ right] {\scriptsize$T(\pi_2^M, \id, \id )$} (F1);
			\draw [->] (F1) -- node [below right] {\scriptsize$T(\id, k,g)$} (B);
			\draw [->] (C) -- node [above] {\scriptsize$f$} (D);
			\draw [->] (E) -- node [above] {\scriptsize$g$} (D);
			\draw [->] (D0) -- node [above] {\scriptsize$\pi_1^R$} (C);
			\draw [->] (D0) -- node [above] {\scriptsize$\pi_2^R$} (E);
			\draw [->] (E) -- node [right] {\scriptsize$\alpha$} (F);
			\draw [->] (D) -- node [right] {\scriptsize$\zeta$} (B);

		\end{tikzpicture}
	\end{center}
	 there exists a morphism $\gamma : \graph{R} \to T(\graph{M}, \graph{N},\graph{R}) $ such that:
	
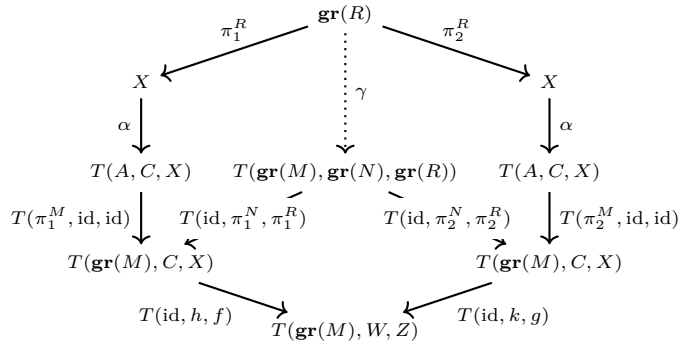
\begin{center}
	\begin{tikzpicture}[thick, xscale=0.9, yscale=0.6]

			\node (A) at (0,0) {\scriptsize$T(A,C,X)$};
			\node (A1) at (0,-2) {\scriptsize$T(\graph{M},C,X)$};
			\node (B) at (3,-3.5) {\scriptsize$T(\graph{M},W,Z)$};
			\node (C) at (0,2) {\scriptsize$X$};
			\node (D0) at (3,3.5) {\scriptsize$\graph{R}$};
			\node (D) at (3,0) {\scriptsize$T(\graph{M},\graph{N},\graph{R})$};
			\node (E) at (6,2) {\scriptsize$X$};
			\node (F) at (6,0) {\scriptsize$T(A,C,X)$};
			\node (F1) at (6,-2) {\scriptsize$T(\graph{M},C,X)$};
			
			\draw [->] (C) -- node [left] {\scriptsize$\alpha$} (A);
			\draw [->] (A1) -- node [below left] {\scriptsize$T(\id, h,f)$} (B);
			\draw [->] (A) -- node [left] {\scriptsize$T(\pi_1^M, \id, \id)$} (A1);
			\draw [->] (F) -- node [ right] {\scriptsize$T(\pi_2^M, \id, \id )$} (F1);
			\draw [->] (F1) -- node [below right] {\scriptsize$T(\id, k,g)$} (B);
			\draw [->] (D) -- node [fill =white] {\scriptsize$T(\id, \pi_1^N, \pi_1^R)$} (A1);
			\draw [->] (D) -- node [fill=white] {\scriptsize$T(\id, \pi_2^N, \pi_2^R)$} (F1);
			\draw [->] (D0) -- node [above] {\scriptsize$\pi_1^R$} (C);
			\draw [->] (D0) -- node [above] {\scriptsize$\pi_2^R$} (E);
			\draw [->] (E) -- node [right] {\scriptsize$\alpha$} (F);
			\draw [->, dotted] (D0) -- node [right] {\scriptsize$\gamma$} (D);

	\end{tikzpicture}
\end{center}
	which finishes the proof. 
\end{proof}

\begin{proof}[Proof of Lemma~\ref{lem:laxdblefunctor}]
	We show both inclusions:
	\begin{align*}
			(T(f,h,i))_*; \Barr{T} (M, N, R) ;  (T(g,k,j))^* &\leq 
			\Barr{T} (f^*,h_*,i_*); \Barr{T} (M, N, R) ;  \Barr{T} (g_*,k^*,j^*) \\
			& \leq \Barr{T} (f^*;M;g_*, h_*;N;k^*, i_*;R;j^*)
	\end{align*}
		\begin{align*}
			&\Barr{T} (f^*;M;g_*, h_*;N;k^*, i_*;R;j^*) \\
			&\leq (T(f,h,i))_*; (T(f,h,i))^*; \Barr{T} (f^*;M;g_*, h_*;N;k^*, i_*;R;j^*) ; (T(g,k,j))_*; (T(g,k,j))^*\\
				&\leq (T(f,h,i))_*; \Barr{T}(f_*,h^*,i^*); \Barr{T} (f^*;M;g_*, h_*;N;k^*, i_*;R;j^*) ; \Barr{T}(g^*,k_*,j_*); (T(g,k,j)^*\\
				&\leq (T(f,h,i))_*; \Barr{T} (f_*;f^*;M;g_*;g^*, h^*;h_*;N;k^*;k_*, i^*i_*;R;j^*;j_*);(T(g,k,j))^*\\
				&\leq 	(T(f,h,i))_*; \Barr{T} (M, N, R) ;  (T(g,k,j))^*\qedhere
		\end{align*}
\end{proof}
\begin{proof}[Proof of Theorem~\ref{th:BehEquivExtFullabstract}]
	\begin{itemize}
		\item[]
		\item 
	Assume that $R$ is an $(M,N)$-behavioral equivalence as in Section~\ref{subsec:paramBehequiv}, then we obtain that it is a $\Barr{T}$-bisimulation w.r.t $(M,N)$:
	\[\begin{aligned}
		f_*; g^* &\leq  	\alpha_* ;  	(T(\pi_1, h,f))_*;(T(\pi_2, k,g))^* ; \alpha^*\\
		& \leq 
		\alpha_* ; 
		\Barr{T}((\pi_1)^*, h_*,f_*);\Barr{T}((\pi_2)_*, k^*,g^*) ; \alpha^*\\
		&=  \alpha_*;	\Barr{T}(M,N,R); \alpha^*
	\end{aligned}
	\]
	\item For the other direction, we want to show that $\Barr{T}$-bisimilarity w.r.t $(M,N)$ is an $(M,N)$-behavioral equivalence.

	The operator $\Rel(X,X) \to \Rel(X,X)$ mapping $R : X\profto X$ to $\alpha_* ; \Barr{T}(M, N, R) ; \alpha^*$ is monotone in $R$ so it has a greatest fixpoint which we denote $E : X\profto X$. If $M \leq \id_{A}$,  $M \leq M \circ M$, $M \leq M^t $, $\id_C \leq N$, $ N \circ N\leq N$, $N^t \leq N$ and $\Barr{T}$ verifies axiom $(6)$ in Definition~\ref{def:ParamLaxExtensionMonadSet}, then $E$ is an equivalence relation (in particular it is difunctional). Since $N$ and $E$ are difunctional, taking the collage of $N$ and $E$, the two squares below are both pushout and weak pullback squares.
	
	
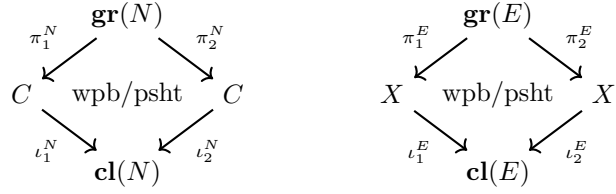
\begin{center}
		\begin{tikzpicture}[thick,scale=0.7]
			\node (A) at (0,1.5) {$C$};
			\node (B) at (2,0) {$\coll{N}$};
			\node (C) at (4,1.5) {$C$};
			\node (D) at (2,3) {$\graph{N}$};
			
			\draw [->] (A)  to  node [below left] {\scriptsize$\iota_1^N$} (B);
			\draw [->] (C)  to  node [below right] {\scriptsize$\iota_2^N$} (B);
			\draw [->] (D) to  node [above left] {\scriptsize$\pi_1^N$} (A);
			\draw [->] (D)  to node [above right] {\scriptsize$\pi_2^N$} (C);
			\node at (2,1.5) {$\mathrm{wpb}/\mathrm{psht}$};
			
			\begin{scope}[xshift=7cm]
				\node (A) at (0,1.5) {$X$};
				\node (B) at (2,0) {$\coll{E}$};
				\node (C) at (4,1.5) {$X$};
				\node (D) at (2,3) {$\graph{E}$};
				
				\draw [->] (A)  to  node [below left] {\scriptsize$\iota_1^E$} (B);
				\draw [->] (C)  to  node [below right] {\scriptsize$\iota_2^E$} (B);
				\draw [->] (D) to  node [above left] {\scriptsize$\pi_1^E$} (A);
				\draw [->] (D)  to node [above right] {\scriptsize$\pi_2^E$} (C);
				\node at (2,1.5) {$\mathrm{wpb}/\mathrm{psht}$};
			\end{scope}
		\end{tikzpicture}
	\end{center}
	In particular, we have $N = (\iota_1^N)_*; (\iota_2^N)^*$ and $E= (\iota_1^E)_*; (\iota_2^E)^*$. By Lemma \ref{lem:laxdblefunctor}, \[\Barr{T}(M, N,R) = (T(\pi_1^M, \iota_1^N, \iota_1^R))_* ; \Barr{T}(\id_A, \id_C, \id_X); (T(\pi_2^M, \iota_2^N, \iota_2^R))^*\] so if $\Barr{T}$ is normal, we have $\Barr{T}(M, N,R) = (T(\pi_1^M, \iota_1^N, \iota_1^R))_* ; (T(\pi_2^M, \iota_2^N, \iota_2^R))^*$. By the universal property of the collage, we obtain that there is a unique function $\coll{E} \to T(\graph{M},\coll{N}, \coll{E})$ such that:
		\begin{center}
		\begin{tikzpicture}[thick,xscale=0.4, yscale =0.5, decoration={ markings, mark=at position 0.5 with {\arrow{|}}}]
			\node (A) at (-3,-3) {\small$T(A,C,X)$};
			\node (A1) at (3,-3) {\small$T(A,C,X)$};
			\node (D) at (0,-6) {\small$T(\graph{M},\coll{N}, \coll{E})$};

			\node (B) at (-3,0) {\small$X$};
			\node (C) at (3,0) {\small$X$};
			
			\draw [->] (B) -- node [ left] {\scriptsize$\alpha$} (A);
			\draw [->] (A) -- node [ left] {\scriptsize$T(\pi_1^M, \iota_1^N, \iota_1^R)$} (D);
			\draw [->] (A1) -- node [right] {\scriptsize$T(\pi_2^M, \iota_2^N, \iota_2^R)$} (D);
			\draw [->, postaction=decorate] (B) -- node [ above] {\scriptsize$E$} (C);
			\draw [->, postaction=decorate] (A) -- node [ above] {\scriptsize$\Barr{T}(M, N,R)$} (A1);
			\draw [-> ] (C) -- node [ right] {\scriptsize$\alpha$} (A1);
			\node 	[rotate=-90] at (0,-1.2) {$\leq$};
			\node 	[rotate=-90] at (0,-4.2) {$\leq$};
			\node at (6,-3) {$=$};
			\begin{scope}[xshift=11cm]
				\node (A) at (0,-6) {\small$T(\graph{M},\coll{N}, \coll{E})$};
				\node (G) at (0,-3) {\small$\coll{E}$};
				\node (B) at (-3,0) {\small$X$};
				\node (C) at (3,0) {\small$X$};
				
				\draw [->, dotted] (G) -- node [ right] {\scriptsize$\zeta$} (A);
				\draw [->] (B) -- node [below left] {\scriptsize$\iota_1^E$} (G);
				\draw [->, postaction=decorate] (B) -- node [ above] {\scriptsize$E$} (C);
				\draw [-> ] (C) -- node [below right] {\scriptsize$\iota_2^E$} (G);
				\node [rotate=-90] at (0,-1.2) {$\leq$};
			\end{scope}
		\end{tikzpicture}
	\end{center}
	In particular, we obtain that $\zeta \circ \iota_1^E = T(\pi_1^M, \iota_1^N, \iota_1^R) \circ \alpha$ and $\zeta \circ \iota_2^E = T(\pi_2^M, \iota_2^N, \iota_2^R) \circ \alpha$ as desired. \qedhere
		\end{itemize}
\end{proof}
\end{document}